\documentclass[a4paper,onecolumn,superscriptaddress,10pt]{marginquantumarticle} 
\pdfoutput=1
\usepackage[activate={true,nocompatibility},final,tracking=true,kerning=true,spacing=true,factor=1100,stretch=10,shrink=10]{microtype}
\usepackage[utf8]{inputenc}
\usepackage[english]{babel}
\usepackage[T1]{fontenc}
\usepackage[numbers,sort&compress]{natbib}
\usepackage{amsmath}
\usepackage{amsthm}
\usepackage{amsbsy}
\usepackage{amsfonts}
\usepackage{enumerate}
\usepackage{bbm}
\usepackage{hhline}
\usepackage{stmaryrd} 
\usepackage{caption}
\usepackage{hyperref}
\usepackage{marginnote}

\usepackage{tabularx}
\newcolumntype{C}[1]{>{\centering\arraybackslash}p{#1}}

\let\oldmarginpar\marginpar
\renewcommand\marginpar[1]{\-\oldmarginpar[\raggedleft\footnotesize #1]%
{\raggedright\footnotesize #1}}
\usepackage{tikz}
\usepackage{lipsum}

\newcommand{\pst}[1]{\overline{#1}}

\newcommand{\scenario}{\mathcal{S}}

\newcommand{\inputs}[1]{\mathcal{I}_{\mathrm{#1}}}
\newcommand{\dd}[1]{{d_{\mathcal{#1}}}}
\newcommand{\RR}{\mathbb{R}}
\newcommand\restr[2]{{
	\left.\kern-\nulldelimiterspace 
	#1 
	\right|_{#2} 
}}

\newcommand{\Mid}{\mathbbm{1}}
\newcommand{\MidP}[1]{\mathbbm{1}_{\mathrm{#1}}}

\newcommand{\domain}[1]{\left\{ 1 .. #1 \right\}}

\newtheorem{corollary}{Corollary}
\newtheorem{lemma}{Lemma}
\newtheorem{definition}{Definition}
\newtheorem{proposition}{Proposition}

\newcommand{\pname}[1]{\mathrm{#1}}

\newcommand{\rvar}[1]{\mathrm{#1}}
\DeclareMathOperator{\spann}{span}
\newcommand{\spanset}[1]{\ensuremath\spann\left(#1\right)}

\newcommand{\Pop}[2]{{\operatorname{P}_{\mathrm{#1}}}\!\left( #2 \right)}
\newcommand{\Pgiven}[2]{{\operatorname{P}_{\mathrm{#1}|\mathrm{#2}}}}
\newcommand{\Pgivenop}[4]{{\operatorname{P}_{\mathrm{#1}|\mathrm{#2}}}\!\left( #3 \middle | #4 \right)}
\newcommand{\Pgivenopprime}[4]{{\operatorname{P'}_{\mathrm{#1}|\mathrm{#2}}}\!\left( #3 \middle | #4 \right)}
\newcommand{\Pnosub}[2]{\operatorname{P}\!\left( #1 \middle | #2 \right)}
\newcommand{\vecP}[1]{{\vec{P}_{\mathcal{#1}}}}
\newcommand{\Pspace}[1]{\mathcal{#1}}
\newcommand{\Pbasis}[3]{{\left\llbracket #2 \middle| #3 \right\rrbracket}_{\mathcal{#1}}}

\newcommand{\Kset}[1]{\mathcal{K}_{\mathcal{#1}}}

\newcommand{\Csub}{\mathsf{C}}
\newcommand{\Zsub}{\mathsf{Z}}
\newcommand{\Ssub}{\mathsf{S}}
\newcommand{\Vsub}{\mathsf{V}}

\newcommand{\Cof}[1]{\mathsf{C}_{\!\mathcal{#1}}}
\newcommand{\Zof}[1]{\mathsf{Z}_{\!\mathcal{#1}}}
\newcommand{\Sof}[1]{\mathsf{S}_{\!\mathcal{#1}}}
\newcommand{\Vof}[1]{\mathsf{V}_{\!\mathcal{#1}}}

\newcommand{\Puni}[1]{\vec{Z}_{\mathcal{#1}}}
\newcommand{\PuniI}[2]{\vec{Z}_{\mathcal{#1}_{#2}}}

\newcommand{\Det}[1]{\mathsf{Det}_{\mathcal{#1}}}

\newcommand{\Lmap}[1]{\boldsymbol{\Lambda}_{\!\mathcal{#1}}}

\newcommand{\id}{\mathsf{id}}

\newcommand{\expr}[1]{\Phi_{\!\mathcal{#1}}}

\newcommand{\traceout}[1]{\tau_{\!\mathcal{#1}}}
\newcommand{\traceoutI}[2]{\tau_{\!\mathcal{#1}_{#2}}}

\newcommand{\Ssum}[1]{\sigma_{\!\mathcal{#1}}}
\newcommand{\SsumI}[2]{\sigma_{\!\mathcal{#1}_{#2}}}
\newcommand{\Sanh}[1]{\overline{\sigma}_{\!\mathcal{#1}}}
\newcommand{\SanhI}[2]{\overline{\sigma}_{\!\mathcal{#1}_{#2}}}

\newcommand{\data}[1]{\vec{C}_{\!\mathcal{#1}}}

\newcommand{\sig}[1]{\vec{S}_{\!\mathcal{#1}}}

\newcommand{\datad}[1]{\gamma_{\!\mathcal{#1}}}
\newcommand{\datadI}[2]{\gamma_{\!\mathcal{#1}_{#2}}}
\usepackage{color,soul}

\begin{document}

\title{Algebraic and geometric properties of local transformations}
\date{\today}
\author{Denis Rosset}
\email{physics@denisrosset.com}
\affiliation{Perimeter Institute for Theoretical Physics, 31 Caroline St. N., Waterloo, Ontario, Canada, N2L 2Y5}
\author{{\"A}min Baumeler}
\affiliation{Institute for Quantum Optics and Quantum Information (IQOQI) Vienna, Austrian Academy of Sciences, Boltzmanngasse 3, 1090 Vienna, Austria}
\affiliation{Faculty of Physics, University of Vienna, Boltzmanngasse 5, 1090 Vienna, Austria}
\affiliation{Facolt\`a indipendente di Gandria, Lunga scala, 6978 Gandria, Switzerland}
\author{Jean-Daniel Bancal}
\author{Nicolas Gisin}
\author{Anthony Martin}
\author{Marc-Olivier Renou}
\affiliation{D\'epartement de Physique Appliqu\'ee, Universit\'e de Gen\`eve, 1211 Gen\`eve, Switzerland}
\affiliation{ICFO-Institut de Ciencies Fotoniques, The Barcelona Institute of Science and Technology, 08860 Castelldefels (Barcelona), Spain}
\author{Elie Wolfe}
\affiliation{Perimeter Institute for Theoretical Physics, 31 Caroline St. N., Waterloo, Ontario, Canada, N2L 2Y5}

\maketitle

\begin{abstract}
  Some properties of physical systems can be characterized from their correlations.
  In that framework, subsystems are viewed as abstract devices that receive measurement settings as inputs and produce measurement outcomes as outputs.
  The labeling convention used to describe these inputs and outputs does not affect the physics; and relabelings are easily implemented by rewiring the input and output ports of the devices.
  However, a more general class of operations can be achieved by using correlated preprocessing and postprocessing of the inputs and outputs.
  In contrast to relabelings, some of these operations irreversibly lose information about the underlying device.
  Other operations are reversible, but modify the number of cardinality of inputs and/or outputs.
  In this work, we single out the set of deterministic local maps as the one satisfying two equivalent constructions: an operational definition from causality, and an axiomatic definition reminiscent of the definition of quantum completely positive trace-preserving maps.
  We then study the algebraic properties of that set.
  Surprisingly, the study of these fundamental properties has deep and practical applications.
  First, the invariant subspaces of these transformations directly decompose the space of correlations/Bell inequalities into nonsignaling, signaling and normalization components.
  This impacts the classification of Bell and causal inequalities, and the construction of assemblages/witnesses in steering scenarios.
  Second, the left and right invertible deterministic local operations provide an operational generalization of the liftings introduced by Pironio [J. Math. Phys.
, 46(6):062112 (2005)].
  Not only Bell-local, but also causal inequalities can be lifted; liftings also apply to correlation boxes in a variety of scenarios.
\end{abstract}

\newpage

\tableofcontents
\newpage

\part*{Contents}
Our motivation is to provide a formal study of the transformations of behaviors in correlation scenarios.
By behavior, we mean a joint conditional probability distribution on devices spanning subsystems.
While such transformations have been studied before~\cite{Barrett2007a, Horodecki2015, Vicente2014}, no detailed study exists that also encompasses scenarios with possible {\em signaling}.
Indeed, scenarios involving signaling directions are increasingly relevant in the study of indefinite causal orders (see {\it e.g.}, Refs.~\cite{Oreshkov2012, Baumeler2016, MacLean2016, Castro-Ruiz2017}).
In this context, we consider two possible definitions of local transformations and show that they single out the same class of maps.
We then study the {\em geometric\/} and {\em algebraic\/} properties of local transformations.
{\em Geometrically}, we show that local transformations decompose the correlation space they act upon into invariant subspaces.
We identify these invariant subspaces with properties such as normalization or signaling, and show that a natural decomposition of the correlation space follows.
{\em Algebraically}, we explore how local transformations compose, and study their invertibility.
We show that invertibility corresponds to the lifting of either behaviors or Bell-like inequalities.

Our manuscript is divided in three parts.
Part~\ref{Part:LocalTransformations} provides the foundations for the rest of the manuscript.
It formally defines scenarios, behaviors, Bell-like inequalities, local transformations and their actions.
In particular, this part singles out the class of local transformations studied in the rest of the work.
Definitions and results are provided in Section~\ref{Part:LocalTransformations} while longer proofs are relegated to Section~\ref{Sec:ProofsLocalTransformations}.

Part~\ref{Part:InvariantSubspaces} studies the invariant subspaces of local transformations: Section~\ref{Sec:InvariantSubspacesSingle} and Section~\ref{Sec:InvariantSubspacesMulti} address the single party and multi-party cases respectively.
We present in Section~\ref{Sec:InvariantApplications} three applications: the equivalency of Bell-like
inequalities under affine transformations and nonsignaling constraints (generalizing the approach of Ref.~\cite{Rosset2014a} to signaling scenarios);
the optimization of the variance of Bell inequalities when used as statistical estimators (generalizing Ref.~\cite{Renou2016})
and the decomposition of assemblages/witnesses in steering scenarios.
Section~\ref{Sec:InvariantTechnical} contains proofs.

Part~\ref{Part:Liftings} studies reversible transformations.
We study composition of local transformations in Section~\ref{Sec:LiftingDefinitions}, before motivating a definition of liftings as generic transformations between equivalent inequalities/behaviors in Section~\ref{Sec:DefiningLiftings}.
In Section~\ref{Sec:LiftingBehaviors}, we consider transformations that create equivalent behaviors from existing behaviors.
We show that a richer class of such transformations exist compared to liftings of inequalities.
We also show that in the nonsignaling scenario where Alice has ternary inputs and outputs, and Bob has binary inputs and outputs, all boxes are either local, or liftings of the PR-box from the CHSH scenario.
In Section~\ref{Sec:LiftingExpressions}, we make an exhaustive inventory of liftings of Bell inequalities, and show that the class of transformations considered by Pironio~\cite{Pironio2005} is complete.
However, our construction applies also to signaling scenarios; we demonstrate that causal inequalities are also affected by lifting redundancies.

\newpage

\part{Local transformations}
\label{Part:LocalTransformations}
In the first part of our manuscript, we formally define the objects under study: scenarios, behaviors, correlation sets, Bell expressions and Bell-like inequalities, and how local transformations act on them.
An excellent preliminary read is the review by Brunner {\it et al.}~\cite{Brunner2014}, as our approach is more mathematical.
The main question we address is the transformation of {\em boxes}, which represent the subsystems in a correlation scenario that possibly includes signaling.
We define local transformations using two approaches: one based on causality (past events cannot depend on future events), and one based on axioms that transformations should obey (such as preserving nonnegativity and normalization).
We show that both definitions single out the same class of local transformations.
In addition, we show that local transformations mirror the positive-but-not-completely-positive property of quantum channels~\cite{Bengtsson2008}.
Here, to show that a local transformation is positive but not completely positive, we will need a signalling distribution (Prop. \ref{Prop:ProperLocalMapsCausal}).
Previous works addressed similar questions.
Barrett~\cite{Barrett2007a} considered normalization-preserving transformations in generalized probabilistic theories.
Due to the nonsignaling constraints, his description of transformations has redundancy; he considers equivalence classes of those and shows that each equivalence class contains a stochastic-like transformation.
The same transformations were studied in greater detail in Ref.~\cite{Horodecki2015}.
These stochastic-like transformations corresponds to the local transformations we study in this Part, although our approach removes the ambiguities.
Another work by de Vincente~\cite{Vicente2014} lists families of local transformations (without claim of exhaustiveness); we recover the operations he lists as subset of our transformations; see also our Part~\ref{Part:Liftings} where we decompose local transformations in detail.

This part is structured as follows.
In Section~\ref{Sec:LocalTransformations:Definitions}, we provide the definitions used in the manuscript: scenarios (Section~\ref{Sec:Definitions:Scenario}), behaviors (or probability distributions, Section~\ref{Sec:Definition:Behaviors}), describe partial or full nonsignaling conditions (Section~\ref{Sec:NonsignalingConditions}), deterministic and local behaviors (Section~\ref{Sec:Def:DeterministicLocalBehaviors}).
This sets up the stage to tackle local transformations (Section~\ref{Sec:Def:LocalMaps}), where the causal and axiomatic definitions are stated to be equivalent in Proposition~\ref{Prop:ProperLocalMapsCausal}.
We then move to characterize such transformations as mixtures (Proposition~\ref{Prop:LocalTransformationsConvex}) of deterministic transformations (Section~\ref{Sec:Def:LocalMaps:Deterministic}).
We define correlation sets closed under local transformations in Section~\ref{Sec:Def:CorrelationSets}; the boundary of such sets is characterized by Bell-like inequalities (Section~\ref{Sec:Def:BellExpressions}), which we decompose as a linear functional associated with an upper bound.
We conclude the Section by demonstrating how Bell-like inequalities transform under local transformations (Proposition~\ref{Prop:BellInequalityTransformation}).
Some proofs of these results were moved to Section~\ref{Sec:ProofsLocalTransformations} to simplify the presentation.

\section{Definitions and preliminary results}
\label{Sec:LocalTransformations:Definitions}
Parties, devices, subsystems or players in a Bell correlation scenario are usually ordered alphabetically as A(lice), B(ob), C(harlie), D(ave) and so on.
We work in the setting where a referee chooses inputs at random and sends them to the parties.
After suitable processing, the parties provide outputs that are collected by the referee who estimates the correlations among the parties.
One can also think of the parties as devices that take measurement settings as input and produce measurement outcomes.
As our description is abstract, we are not concerned by these details and use the terminology party/input/output.

\subsection{Scenario}
\label{Sec:Definitions:Scenario}
A scenario is composed of a number of devices labeled A, B, \ldots.
Most of our definitions and results are stated in the two-party case;
except when explicitly mentioned, the multi-party generalization is straightforward.
The devices receive inputs taken from finite sets; without loss of generality,
the device A receives $x \in \domain{X}$, while the device B receives~$y \in \domain{Y}$.
The integers $X$ and $Y$ are the numbers of input values.
The devices' outputs also have finite cardinality, but their number can depend on the input.
When the device A receives the input $x$, it outputs $a \in \domain{A_x}$; respectively for the input $y$ the device B outputs $b \in \domain{B_Y}$.
The {\em cardinality} of the party/device A is given by the sequence $\pst{A} = (A_1, \ldots, A_X)$ and similarly for $\pst{B}$.
When necessary, we will use additional devices C and D with inputs $z$ and $t$, and outputs $c$ and $d$, the rest of the notation being easily deduced.

\marginnote{For example, a two-party
	nonsignaling scenario has $E^\text{NS} = \{ (\text{A}, \text{A}), (\text{B}, \text{B}) \}$, whereas the two-party scenario used to study indefinite causal orders in~\cite{Branciard2016} has $E = E^\text{NS} \cup \{ (\text{A}, \text{B}), (\text{B}, \text{A}) \}$.
}
Depending on the underlying causal structure, there will be restrictions on the correlations between inputs and outputs of distinct parties, for example due to the impossibility of faster-than-light communication.

The {\em signaling directions} of a scenario are described by the set of pairs $E = \{ (s, t) \}$, where $s, t \in \{ \text{A}, \text{B}, \ldots \}$ and $s$ can signal to $t$.
For simplicity, we define that $(s, s) \in E$ for all $s$.

\begin{definition}
A {\em scenario} $\scenario$ is defined by the cardinality of its parties $(\pst{A}, \pst{B}, \ldots)$ and the signaling directions $E$. 
\end{definition}

The interpretation of those signaling directions is clarified in Section~\ref{Sec:NonsignalingConditions}.

\subsection{Behaviors}
\label{Sec:Definition:Behaviors}

\marginnote{We recall that the Kronecker product is defined, for $\vec{x} \in \RR^n$ and $\vec{y} \in \RR^m$ as
  \[
    \vec{x} \otimes \vec{y} = \begin{pmatrix} x_1 \vec{y} \\ \ldots \\ x_n \vec{y} \end{pmatrix}\;.
  \]
}
The behavior of devices is fully described by the distribution $\Pgivenop{AB}{XY}{ab}{xy}$.
We enumerate the coefficients of that distribution in a column vector $\vecP{AB} \in \Pspace{A} \otimes \Pspace{B}$.
More precisely, for $d_\Pspace{A} = \sum_x A_x$, vectors in the space $\Pspace{A} = \RR^{d_\Pspace{A}}$ correspond to the enumeration of the coefficients of the single party distribution $\Pgivenop{A}{X}{a}{x}$ obtained by first incrementing the index $a$ and then $x$.
The same holds for $\Pspace{B}$ and for the spaces of the subsequent parties.
We fix the enumeration in $\vecP{AB}$ by requiring that the coefficient order in $\vecP{AB}$ corresponds to the Kronecker product $\vecP{A} \otimes \vecP{B}$ when $\Pgivenop{AB}{XY}{ab}{xy} = \Pgivenop{A}{X}{a}{x} \Pgivenop{B}{Y}{b}{y}$.

As $\Pspace{A}$, $\Pspace{B}$ and $\Pspace{A} \otimes \Pspace{B}$ are vector spaces, they include elements that do not correspond to proper probability distributions.
We define their nonnegative subsets
\begin{equation}
  \label{Eq:NonnegativeSubset}
  (\Pspace{A} \otimes \Pspace{B})^+ = \Big \{ \vecP{AB} \in \Pspace{A} \otimes \Pspace{B} \quad  : \quad \forall a,b,x,y, ~ \Pgivenop{AB}{XY}{ab}{xy} \ge 0 \Big \}
  \;,
\end{equation}
while the normalized subset is
\begin{equation}
  \label{Eq:NormalizedSubset}
  (\Pspace{A} \otimes \Pspace{B})^\Sigma = \Big \{ \vecP{AB} \in \Pspace{A} \otimes \Pspace{B} \quad : \quad \forall x,y, ~ \sum_{ab} \Pgivenop{AB}{XY}{ab}{xy} = 1 \Big \}\;.
\end{equation}
Combining those two definitions, we denote by $(\Pspace{A} \otimes \Pspace{B})^{\Sigma+} = (\Pspace{A} \otimes \Pspace{B})^{\Sigma} \cap (\Pspace{A} \otimes \Pspace{B})^{+}$ the set of normalized, nonnegative behaviors.
Those definitions were made for two-party behaviors; but similar definitions apply to single party behaviors (for example $\Pspace{A}^+$, $\Pspace{A}^\Sigma$).

\subsection{Nonsignaling conditions}
\label{Sec:NonsignalingConditions}
We now describe the nonsignaling constraints~\cite{Brunner2014} for two-party distributions.
Without loss of generality, we consider the case where B does not signal to C, that is $(B, C) \notin E$.
The behavior $\Pgivenop{BC}{YZ}{bc}{yz}$ then obeys the nonsignaling constraint:
\begin{equation}
  \sum_{b} \Pgivenop{BC}{YZ}{bc}{yz} = \sum_{b} \Pgivenop{BC}{YZ}{bc}{y'z}, \qquad \forall c,y,y',z\;.
\end{equation}

Let us now consider the multi-party case.
\marginnote{The multi-party definition is relevant to understand the decomposition of the correlation space into nonsignaling and signaling subspaces in Part~\ref{Part:InvariantSubspaces}, and can be skipped at first reading.}
To simplify the notation in the definition below, and in a few other parts of the manuscript, we temporarily group the parties into sets such as $\{ A_1, A_2 \ldots\}$ depending on their role in the nonsignaling conditions.
\begin{definition}
  \label{Def:Nonsignaling}
  In a given scenario, we consider all subsets of parties that obey the condition below, where we relabel the parties for convenience, and enumerate the corresponding {\em nonsignaling constraints}.
  We consider a source subset of parties $\{ \pname{B}_1, \pname{B}_2, \ldots \}$, and a target subset $\{ \pname{C}_1, \pname{C}_2, \ldots \}$ such that no source signals to a target: $(\pname{B}_i, \pname{C}_j) \notin E$.
  The remaining parties are enumerated $\{ \pname{A}_1, \pname{A}_2, \ldots \}$.
  To simplify the notation in the equation below, we regroup the variables $\overline{\rvar{A}} = (\rvar{A}_1, \rvar{A}_2, \ldots)$, $\overline{\rvar{X}} = (\rvar{X}_1, \rvar{X}_2, \ldots)$ and the indices $\overline{a} = (a_1, a_2, \ldots)$, $\overline{x} = (x_1, x_2, \ldots)$, and similarly for the two other sets of parties.
  
  The behaviors of that scenario obey the constraint:
  \begin{equation}
    \sum_{\overline{a}\overline{b}} \Pgivenop{\overline{A}\overline{B}\overline{C}}{\overline{X}\overline{Y}\overline{Z}}{\overline{a} \overline{b} \overline{c}}{\overline{x} \overline{y} \overline{z}} = \sum_{\overline{a}\overline{b}} \Pgivenop{\overline{A}\overline{B}\overline{C}}{\overline{X}\overline{Y}\overline{Z}}{\overline{a} \overline{b} \overline{c}}{\overline{x} \overline{y}' \overline{z}}, \quad \forall \overline{c},\overline{x},\overline{y}, \overline{y}', \overline{z}\;.
  \end{equation}
\end{definition}

We now revert to the original enumeration A, B, C, ... of the parties.

\subsection{Deterministic and local behaviors}
\label{Sec:Def:DeterministicLocalBehaviors}

Let $\inputs{A}$ be the inputs of the parties that can signal to A, with A itself included (i.e. $x \in \inputs{A}$ always):
\begin{equation}
  \inputs{A} = \{ \mathrm{input}(P) : (P, A) \in E \}\;.
\end{equation}
and the same for $\inputs{B}$ and B, and so on.
Then, a {\em deterministic behavior} is written
\begin{equation}
  \Pgivenop{AB\ldots}{XY\ldots}{ab\ldots}{xy\ldots} = \Pgivenop{A}{\inputs{A}}{a}{\inputs{A}} ~ \Pgivenop{B}{\inputs{B}}{b}{\inputs{B}} ~ \ldots \;,
\end{equation}
where $\Pgiven{A}{\inputs{A}}$, $\Pgiven{B}{\inputs{B}}$, \ldots are deterministic distributions with coefficients in $\{0,1\}$.

It is straightforward to verify that such behaviors obey the nonsignaling conditions of Definition~\ref{Def:Nonsignaling}.
We define now local behaviors.
\begin{definition}
  \label{Def:LocalBehaviors}
  A {\em local behavior} is a convex mixture of deterministic behaviors.
\end{definition}
By linearity, local behaviors obey the nonsignaling conditions of Definition~\ref{Def:Nonsignaling}.

\subsection{Local transformations}
\label{Sec:Def:LocalMaps}
Consider a device A with cardinality $\pst{A}$.
We can apply processing before and after operating the device: for example preprocess its input or postprocess its output.
These extra operations can even be correlated.
As a consequence of this processing, the resulting device can have a different structure $\pst{A}'$, for example with additional inputs or outputs.
As we will see, such transformations can be used to adapt any device to a given structure, possibly losing information in the process.

\marginnote{
  A few examples for $\pst{A}=(2,2)$ and maps $\Pspace{A} \to \Pspace{A}$.
  The permutation of inputs is written
  \begin{equation*}
    \Lmap{A}^\text{IF} = \begin{pmatrix} 0 & 0 & 1 & 0 \\ 0 & 0 & 0 & 1 \\ 1 & 0 & 0 & 0 \\ 0 & 1 & 0 & 0 \end{pmatrix}\;,
  \end{equation*}
  while output randomization is
  \begin{equation*}
    \Lmap{A}^\text{RND} = \frac{1}{2} \begin{pmatrix} 1 & 1 & 1 & 1 \\ 1 & 1 & 1 & 1 \\ 1 & 1 & 1 & 1 \\ 1 & 1 & 1 & 1 \end{pmatrix}\;.
  \end{equation*}
  }
Formally, given $\vecP{A}\in\Pspace{A}$, we look for a map $\Lmap{A}: \Pspace{A}\to\Pspace{A}'$ such that $\vecP{A'}' = \Lmap{A} \vecP{A}$.
In our probabilistic setting, the map $\Lmap{A}$ has to be linear to preserve convexity.
Hence, we write $\Lmap{A}$ as a matrix in $\RR^{{\dd{A'}}\times\dd{A}}$.
The row space of $\Lambda$ is implicitly indexed by $(a',x')$, while its column space is indexed by $(a,x)$.
The probability distribution $\vecP{A}$ transforms to $\vecP{A'}'$:
\begin{equation}
  \Pgivenop{A'}{X'}{a'}{x'} = \sum_{x=1}^X \sum_{a=1}^{A_x} \Lambda_{(a',x'),(a,x)} \Pgivenop{A}{X}{a}{x} \;.
\end{equation}

This notation is compatible with the tensor product structure; for example, when applying $\Lmap{A}$ on the party A but leaving B untransformed, we write:
\begin{equation}
  \vecP{A'B}' = (\Lmap{A} \otimes \MidP{B}) \cdot \vecP{AB}\;, \quad \Pgivenopprime{A'B}{X'Y}{a'b}{x'y} = \sum_{ax} \Lambda_{(a',x'),(a,x)} \Pgivenop{AB}{XY}{ab}{xy}\;,
\end{equation}
with $\MidP{B}: \Pspace{B} \to \Pspace{B}$ the identity map, and $\vecP{A'B}' \in \Pspace{A}' \otimes \Pspace{B}$.

Not all matrices $\Lmap{A}$ correspond to a sound transformation.

\subsubsection{Causal and axiomatic definitions}
\label{Sec:Def:LocalMaps:Definitions}
The subset of local maps can be defined in two equivalent ways, which we investigate below.
We can first ask that any processing should follow causality.
For example, the postprocessing of outputs can depend on the input but not the other way around.
\marginpar{See also the definition of 1W-LOCC transformations in Refs.~\cite{Gallego2015,Kaur2017a}.}
\begin{definition}[Causal local transformations]
  \label{Def:CausalLocalTransformations}
  Causal local maps are composed of an input preprocessing step $\Pgivenop{X}{X'}{x}{x'}$ and an output postprocessing step $\Pgivenop{A'}{XAX'}{a'}{x a x'}$ such that
  \begin{equation}
    \label{Eq:CausalLocalTransformations}
    \Pgivenop{A'}{X'}{a'}{x'} = \sum_{ax} \underbrace{\Pgivenop{X}{X'}{x}{x'} \Pgivenop{A'}{XAX'}{a'}{x a x'}}_{\Lambda_{(a',x'),(a,x)}} \Pgivenop{A}{X}{a}{x}\;,
  \end{equation}
where  $\Pgivenop{X}{X'}{x}{x'}$ and $\Pgivenop{A'}{XAX'}{a'}{x a x'}$ are probability distributions.
\end{definition}

The second way is to define axioms that local transformations should obey.
For example, the processing should preserve normalization and the nonnegativity of coefficients, even when applied in arbitrary multi-party scenarios.
\marginnote{Note that all the maps we consider are normalization preserving.
  This is in contrast with the quantum case, where CP maps are not necessarily completely positive and trace preserving (CPTP).
  Our definitions and proofs still apply when requiring a weaker condition: Behaviors can be subnormalized by a factor that is constant over all input combinations; and that positive/completely positive maps preserve the consistency of subnormalization across inputs, but can modify that factor.
}
\begin{definition}[Positive and completely positive (CP) local transformations]
  \label{Def:PositiveTransformations}
  A local transformation~$\Lmap{A}:\Pspace{A}\to\Pspace{A'}$ is called {\em positive\/} if and only if it maps all normalized, positive behaviors to normalized, positive behaviors, {\it i.e.},
  \begin{equation}
    \vecP{A}\in\Pspace{A}^{\Sigma+} \quad \Rightarrow \quad \Lmap{A} ~ \vecP{A} \in \Pspace{A}^{\Sigma+}
    \,.
  \end{equation}
  
  A local transformation~$\Lmap{A}:\Pspace{A}\to\Pspace{A'}$ is called {\em completely positive\/} if and only if it maps all joint normalized, positive behaviors~$A,B$ to joint normalized, positive behaviors~$A',B$ via partial application on the first, {\it i.e.},
  \begin{align}
    \vecP{AB}\in(\Pspace{A}\otimes\Pspace{B})^{\Sigma+} \quad \Rightarrow \quad (\Lmap{A} \otimes \MidP{B}) \cdot \vecP{AB} \in (\Pspace{A'} \otimes \Pspace{B})^{\Sigma+}
    \,,
  \end{align}
  for all cardinalities $\pst{B}$ and all scenarios, including those with signaling from A to B.
\end{definition}

\begin{proposition}
  \label{Prop:ProperLocalMapsCausal}
  Causal local transformations (Definition~\ref{Def:CausalLocalTransformations}) are equivalent to completely positive local transformations (Definition~\ref{Def:PositiveTransformations}).
\end{proposition}
\begin{proof}
  Causal transformations are completely positive by the rules of probability theory.
  The converse is proven in Section~\ref{Sec:ProofsLocalTransformations:Causal}.
\end{proof}
We provide here an example that illustrates that complete positivity is required.
Let $\Lmap{A}: \Pspace{A} \to \Pspace{A}$ with $\pst{A} = (1,1)$; the device A provides a choice of two inputs but always returns the same outcome.
It is clear that $\vecP{A} = (1,1)^\top$ is the only normalized behavior corresponding to that device.
The following map $\Lmap{A}$ is positive and preserves normalization.
Actually, it leaves the only possible $\vecP{A}$ invariant:
\begin{equation}
  \Lmap{A} = \begin{pmatrix} 2 & -1 \\ -1 & 2 \end{pmatrix}\;.
\end{equation}
Nevertheless $(\Lmap{A} \otimes \MidP{B})$ fails to preserve nonnegativity when applied on the signaling behavior
\begin{equation}\label{Eq:signalling}
  \Pgivenop{AB}{XY}{ab}{xy} = \begin{cases} 1 & \mbox{if } b = x\;, \\ 0 & \mbox{otherwise,} \end{cases}
\end{equation}
with the device B having a single input with binary outputs $\pst{B} = (2)$.

We move towards the algebraic characterization of local maps.

\subsubsection{Deterministic local maps}
\label{Sec:Def:LocalMaps:Deterministic}
\marginnote{Deterministic local transformations are studied in greater details in Section~\ref{Sec:LifitingDefinitions:DeterministicLocalTransformations}.}
We now consider transformations of the form~\eqref{Eq:CausalLocalTransformations} where
$\Pgiven{X}{X'}$ and $\Pgiven{A'}{XAX'}$ are deterministic.
As $x$ is fully determined by $x'$, it is sufficient to consider local maps of the form
\begin{equation}
  \Lambda_{(a',x'),(a,x)} = \Pgivenop{A'}{AX'}{a'}{ax'} \Pgivenop{X}{X'}{x}{x'}
\end{equation}
with deterministic distributions $\Pgiven{A'}{AX'}$ and $\Pgiven{X}{X'}$, as pictured on Figure~\ref{Fig:DeterministicLocalMap}.
We provide now a compact notation for those deterministic local maps.
\marginnote{
  \begin{center}
    \includegraphics{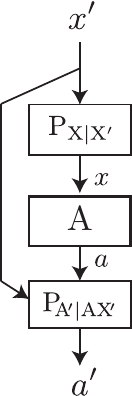}
  \end{center}
  \captionof{figure}{\label{Fig:DeterministicLocalMap}
    General form of a deterministic local map.
  }
}
\begin{definition}
  \label{Def:LocalDeterministicMap}
  A {\em local deterministic map} $\Lmap{A}$ is fully determined by the mapping of inputs
  \begin{equation}
    \xi: \domain{X'} \to \domain{X}, \qquad \xi: x' \mapsto x \;,
  \end{equation}
  and the mapping of outputs $\overline{\alpha} = (\alpha_1, \dots, \alpha_{X'})$, eventually conditioned on $x'$:
  \begin{equation}
    \alpha_{x'}: \domain{A_{\xi(x')}} \to \domain{A'_{x'}}, \qquad \alpha_{x'}: a \mapsto a' \;,
  \end{equation}
  so that
  \begin{equation}
    \Lambda_{(a',x'),(a,x)} = \begin{cases} 1 & \mbox{if }  x = \xi(x') \mbox{ and }  a' = \alpha_{x'}(a)\;, \\ 0 & \mbox{otherwise.} \end{cases}
  \end{equation}
\end{definition}

\subsubsection{All local transformations}
We now arrive at our main characterization.

\begin{proposition}
  \label{Prop:LocalTransformationsConvex}
  All causal (or, equivalently, completely positive) local transformations can be written as a convex mixture of deterministic local maps
  \begin{equation}
    \boxed{
      \Lmap{A} = \sum_i p_i ~ \Lmap{A}^i
    }
  \end{equation}
  where $\sum_i p_i = 1$, $p_i \ge 0$ and the $\Lmap{A}^i$ satisfy Definition~\ref{Def:LocalDeterministicMap}.
\end{proposition}
\begin{proof}
  See Section~\ref{Sec:ProofsLocalTransformations:Convex}.
\end{proof}

\subsection{Correlation sets}
\label{Sec:Def:CorrelationSets}
\marginnote{
 \begin{center}
   \includegraphics{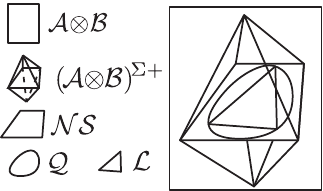}
 \end{center}
 \captionof{figure}{
   \label{Fig:CorrelationSets}
   The vector space $\Pspace{A} \otimes \Pspace{B}$ contains normalized and nonnegative probability distributions $(\Pspace{A} \otimes \Pspace{B})^{\Sigma+}$.
   In quantum information, we customarily distinguish the subspace $\mathcal{NS}$ obeying the nonsignaling conditions of Definition~\ref{Def:Nonsignaling}, the set of quantum correlations~\cite{Goh2018} $\mathcal{Q}$ and the local set $\mathcal{L}$ of Definition~\ref{Def:LocalBehaviors}.
 }
}
The set of local behaviors is closed under local transformations.
\begin{proposition}
  Let $\vecP{AB}$ be a local behavior according to Definition~\ref{Def:LocalBehaviors}, and $\Lmap{A}$, $\Lmap{B}$ be local transformations according to Proposition~\ref{Prop:LocalTransformationsConvex}.
  Then $(\Lmap{A} \otimes \Lmap{B}) ~ \vecP{AB}$ is a local behavior.
\end{proposition}
\begin{proof}(Sketch).
  Remark that a deterministic local transformation applied to a deterministic behavior produces a deterministic behavior.
  The result follows from the convex decomposition of the local behavior into deterministic behaviors, and the convex decomposition of the local transformation into deterministic local transformations. 
\end{proof}

Many types of correlation sets are convex, in particular when the underlying scenario allows for shared randomness between all parties.
We are particularly interested in the convex correlation sets that are closed under local transformations.
This is the case for the set of quantum correlations~\cite{Goh2018} or almost quantum correlations~\cite{Navascues2015}.
Some convex correlation sets are polytopes.
For
example, the local, nonsignaling~\cite{Brunner2014} and causally ordered sets~\cite{Branciard2016} are polytopes.

By the hyperplane separation theorem~\cite[Corollary 11.4.2]{Rockafellar1970}, if a correlation vector $\vecP{AB}$ is not part of some convex correlation set $\Kset{AB}$, there exists a linear inequality that separates that vector from the set:
\begin{equation}
  \label{Eq:SeparatingHyperplane}
  \sum_{abxy} \phi(a,b,x,y) ~ \Pgivenop{AB}{XY}{ab}{xy} \le u \;,
\end{equation}
a property we formalize below.

\subsection{Bell expressions and Bell-like inequalities}
\label{Sec:Def:BellExpressions}
\marginnote{
  It would be a mistake to think of $\expr{AB}$ and $\vecP{AB}$ as members of the {\em same} vector space. 
  Bell expressions and probability vectors do not transform alike.
  Their invariant spaces are different, as Propositions~\ref{Prop:Decomposition} and~\ref{Prop:DecompositionDual} show.
  Moreover, the ``inner products'' $\expr{AB} \cdot \expr{AB}'$ or $\vecP{AB} \cdot \vecP{AB}'$ would have little physical justification.
}[-1cm]
A {\em Bell expression} is a linear map (or linear form) $\expr{AB}: \Pspace{A} \otimes \Pspace{B} \to \RR$.
Formally, $\expr{AB}$ is an element of the dual vector space $(\Pspace{A} \otimes \Pspace{B})^*$.
If we identify probability vectors $\vecP{AB}$ with column vectors, then Bell expressions are row vectors with real coefficients $\phi(a,b,x,y)$ so that
\begin{equation}
  \expr{AB} ~ \vecP{AB} = \sum_{abxy} \phi(a,b,x,y) ~ \Pgivenop{AB}{XY}{ab}{xy} \;.
\end{equation}
We use the following convention: We write column vectors/behaviors using latin letters with a vector arrow (as in $\vecP{}$), row vectors/Bell expressions using greek letters without an arrow (as in $\expr{}$), and matrices/linear maps using bold greek letters (as in $\Lmap{}$).

\marginnote{
 \begin{center}
   \includegraphics{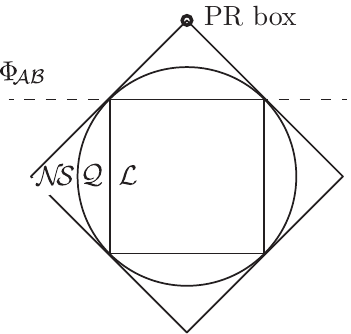}
 \end{center}
 \captionof{figure}{
   \label{Fig:KInequality}
   The CHSH inequality is also a $\mathcal{L}$-inequality certifying that the PR box correlations are not part of the local set $\mathcal{L}$.
 }
}[-2cm]
Formally, the membership certificates presented in~\eqref{Eq:SeparatingHyperplane} are defined as follows (see example in Figure~\ref{Fig:KInequality}).
\begin{definition}
  A $\Kset{}${\em-inequality} $(\expr{AB}, u)$ is a Bell expression $\expr{AB}\in(\Pspace{A}\otimes\Pspace{B})^*$ along with an upper bound $u \in \mathbb{R}$ such that:
  \begin{equation}
    \vecP{AB} \in \Kset{AB} \quad \Rightarrow \quad \expr{AB} \vecP{AB} \le u\;,
  \end{equation}
  where~$\Kset{AB}$ is a convex set closed under local transformations.
\end{definition}

\subsubsection{Transformations of Bell expressions}
Remember that local transformations act on behaviors as matrix-vector multiplication.
For $\vecP{A} \in \Pspace{A}$, $\Lmap{A}: \Pspace{A} \to \Pspace{A'}$:
\begin{equation}
  \vecP{A'}' = \Lmap{A} ~ \vecP{A}, \qquad \Pgivenop{A'}{X'}{a'}{x'} = \sum_{ax} \Lambda_{(a',x),(a,x)} ~ \Pgivenop{A}{X}{a}{x}\;.
\end{equation}

We can also define an action of local transformations on Bell expressions.
\marginnote{Note that the order of source and target spaces is reversed between $\Lmap{A}$ and $\Lmap{A'}$.}
Let $\Lmap{A'}: \Pspace{A'} \to \Pspace{A}$ be a local transformation, and $\expr{A}: \Pspace{A} \to \RR$ a Bell expression.
Then $\expr{A'}' = \expr{A} ~ \Lmap{A'}$ defines a Bell expression that takes a behavior in $\Pspace{A'}$, applies the transformation $\Lmap{A'}$ to obtain a behavior in $\Pspace{A}$ and finally evaluates the original $\expr{A}$ on the transformed behavior.
This corresponds to the row-vector-matrix multiplication
\begin{equation}
  \phi'(a',x') = \sum_{ax} \phi(a,x) ~ \Lambda'_{(a,x),(a',x')} \;.
\end{equation}
\marginnote{
Formally~\cite[Chap. 3]{Roman2005}, this action of local transformation corresponds to the adjoint of $\Lmap{A'}: \Pspace{A'} \to \Pspace{A}$, usually written $\Lmap{A'}^\dagger: \Pspace{A}^* \to \Pspace{A'}^*$.
}
We easily verify the following proposition.
\begin{proposition}
  \label{Prop:BellInequalityTransformation}
  Let $(\expr{AB}, u)$ be a $\Kset{}$-inequality with $\expr{AB} \in (\Pspace{A} \otimes \Pspace{B})^*$.
  Let $\Lmap{A'}: \Pspace{A'} \to \Pspace{A}$ and $\Lmap{B'}: \Pspace{B'} \to \Pspace{B}$ be local transformations.
  Then $(\expr{A'B'}', u)$ with $\expr{A'B'}' = \expr{AB} (\Lmap{A'} \otimes \Lmap{B'})$ is also a $\Kset{}$-inequality.
\end{proposition}
\begin{proof}
  Shift the application of $(\Lmap{A'} \otimes \Lmap{B'})$ to $\vecP{A'B'}$ and note that, by definition, the correlation set $\Kset{AB}$ is closed under local transformations.
\end{proof}

\section{Proofs}
\label{Sec:ProofsLocalTransformations}

Here we provide proofs of the preceding propositions.

\subsection{Completely positive local transformations are causal}
\label{Sec:ProofsLocalTransformations:Causal}
We prove Proposition~\ref{Prop:ProperLocalMapsCausal} in two steps: First we show that CP local maps are conditional probability distributions, then we prove the proposition, {\it i.e.}, CP local maps can be understood as a combination of pre- and post-processing.
\begin{lemma}
  CP local maps~$\Lmap{A}:\Pspace{A}\rightarrow\Pspace{A}'$ correspond to conditional probability distributions of the form~$\Pgiven{A'X}{AX'}$, {\it i.e.},
  \begin{equation}
    \exists \Pgiven{A'X}{AX'}: \Lambda_{(a',x'),(a,x)}=\Pgivenop{A'X}{AX'}{a',x}{a,x'}
    \,.
  \end{equation}
\end{lemma}
\begin{proof}
  Let~$\Lmap{A}$ act on the first part of some joint device~$\vecP{AB}\geq0$ which is defined as follows:
  \begin{align}
    &\pst{A}=(A_1,\dots,A_X)\,,\notag\\
    &\pst{B}=(B_1,\dots,B_X)\quad\text{with}\quad B_i=\max_x A_x\,,\notag\\
    &\Pgivenop{AB}{XY}{a,b}{x,y}=
      \delta_{a,\restr{y}{A_x}} \delta_{b,x}\,,
      \label{eq:swap}
  \end{align}
  where~$\restr{y}{A_x}=((y-1)\mod A_x)+1$.
  This device takes takes two inputs and swaps them, where additionally the output~$a$ is truncated to the range~$1,\dots,A_x$ appropriate for input~$x$.
  The distribution~$\Pgiven{AB}{XY}$ transforms to~$\Pgiven{A'B}{X'Y}$ via
  \begin{align}
    \Pgivenop{A'B}{X'Y}{a',b}{x',y}=\sum_{x=1}^X\sum_{a=1}^{A_x}\Lambda_{(a',x'),(a,x)}\Pgivenop{AB}{XY}{a,b}{x,y}
    \,.
    \notag
  \end{align}
  Since~$\Lmap{A}$ is completely positive, we have
  \begin{align}
    \forall a',b,x',y:& \Pgivenop{A'B}{X'Y}{a',b}{x',y}\geq 0\,,\notag\\
    \forall x',y:& \sum_{a'=1}^{A_{x'}}\sum_{b=1}^{B_y} \Pgivenop{A'B}{X'Y}{a',b}{x',y} =1
                   \,.
                   \notag
	\end{align}
	By plugging Equation~\eqref{eq:swap} into the above equations, we obtain that the following is nonnegative
	for every choice of~$a',b,x',y$:
	\begin{align}
          &\sum_{x=1}^X\sum_{a=1}^{A_x}\Lambda_{(a',x'),(a,x)}
            \delta_{a,\restr{y}{A_x}} \delta_{b,x}
            = \Lambda_{(a',x'),(\restr{y}{A_b},b)}
		\,,
		\notag
	\end{align}
	and that the following sums to~$1$ for every choice of~$x',y$:
	\begin{align}
		\sum_{a'=1}^{A_{x'}}
		\sum_{b=1}^{B_y}
		\sum_{x=1}^X
		\sum_{a=1}^{A_x}
		\Lambda_{(a',x'),(a,x)}
		\delta_{a,\restr{y}{A_x}} \delta_{b,x}
		=
		\sum_{a'=1}^{A_{x'}}
		\sum_{b=1}^{B_y}
		\Lambda_{(a',x'),(\restr{y}{A_b},b)}
		\,.
		\notag
	\end{align}
	So,~$\Lmap{A}$ is a nonnegative function on four variables with the property that the sum over~$a'$ and~$b$ yields 1: It is a probability distribution of the form~$\Pgiven{A'X}{AX'}$.
\end{proof}
By invoking the above Lemma, we prove Proposition~\ref{Prop:ProperLocalMapsCausal}.
\begin{proof}[Proof of Proposition~\ref{Prop:ProperLocalMapsCausal}]
	Suppose towards a contradiction that there exist some values~$x'_0,x_0,a_0,a_+$ where
	\begin{align}
		S:=\sum_{a'}\Pgivenop{A'X}{AX'}{a',x_0}{a_0,x'_0} < \sum_{a'}\Pgivenop{A'X}{AX'}{a',x_0}{a_+,x'_0}=:Q
		\,,
		\notag
	\end{align}
	{\it i.e.}, the probability that~$\Lmap{A}$ outputs~$x_0$ depends on the value of~$A$.
	If we would apply this CP local map on the device
	\begin{align}
		\Pgivenop{A}{X}{a}{x}=\delta_{x,x_0}\delta_{a,a_0}+(1-\delta_{x,x_0})\delta_{a,a_+}
		\,,
		\notag
	\end{align}
	then normalization is {\em not\/} preserved:
	\begin{align}
		\sum_{a',x,a}&\Pgivenop{A}{X}{a}{x}\Pgivenop{A'X}{AX'}{a',x}{a,x'_0}\notag\\
		&=
		\sum_{a',x,a}\left(\delta_{x,x_0}\delta_{a,a_0}+(1-\delta_{x,x_0})\delta_{a,a_+}\right) \Pgivenop{A'X}{AX'}{a',x}{a,x'_0}\notag\\
		&=
		\sum_{a'}\Pgivenop{A'X}{AX'}{a',x_0}{a_0,x'_0}+\sum_{a',x\not=x_0}\Pgivenop{A'X}{AX'}{a',x}{a_+,x'_0}\notag\\
		&=
		S+1-Q<1
		\,.
		\notag
	\end{align}
\end{proof}

\subsection{Local transformations as convex mixtures of deterministic transformations}
\label{Sec:ProofsLocalTransformations:Convex}
Thanks to Proposition~\ref{Prop:ProperLocalMapsCausal} we can identify the CP local map~$\Lmap{A}$ with~$\Pgiven{X}{X'}\Pgiven{A'}{XAX'}$ and~$\Lmap{A}^i$ with~$\Pgiven{X}{X',I=i}\Pgiven{A'}{AX',I=i}$.
Thus, in order to prove Proposition~\ref{Prop:LocalTransformationsConvex}, we have to show that~$\Pgiven{X}{X'}\Pgiven{A'}{XAX'}$ can be written as a convex combination of {\em deterministic\/} distributions~$\Pgiven{X}{X',I=i}\Pgiven{A'}{AX',I=i}$.
\begin{proof}[Proof of Proposition~\ref{Prop:LocalTransformationsConvex}]
	First, we decompose~$\Pgiven{X}{X'}$ and~$\Pgiven{A'}{XAX'}$ as convex combinations of deterministic distributions:
	\begin{align}
		\Pgivenop{X}{X'}{x}{x'}&=\sum_{\xi}\Pop{\Xi}{\xi}\delta_{x,f_\xi(x')}\,,\notag\\
		\Pgivenop{A'}{XAX'}{a'}{x,a,x'}&=\sum_{\omega}\Pop{\Omega}{\omega}\delta_{a',g_\omega(x,a,x')}\notag\,,
	\end{align}
	where~$f_\xi$ and~$g_\omega$ are functions.
	From this, we get
	\begin{align}
		\Pgivenop{X}{X'}{x}{x'}&\Pgivenop{A'}{XAX'}{a'}{x,a,x'}\notag\\
		&= \sum_{\xi\omega}\Pop{\Xi}{\xi}\Pop{\Omega}{\omega} \delta_{x,f_\xi(x')} \delta_{a',g_\omega(x,a,x')}\notag\\
		&= \sum_{\xi\omega}\Pop{\Xi}{\xi}\Pop{\Omega}{\omega} \delta_{x,f_\xi(x')} \delta_{a',g_\omega(f_\xi(x'),a,x')}\notag\,.
	\end{align}
	Now, we define a random variable~$I=\Xi\times\Omega$:
	\begin{align}
		\Pop{I}{i}=\Pop{I}{\xi,\omega}:=\Pop{\Xi}{\xi}\Pop{\Omega}{\omega}
		\,;
		\notag
	\end{align}
	and a function
	\begin{align}
		g'_i(a,x'):=g_\omega(f_\xi(x'),a,x')
		\,.
		\notag
	\end{align}
	We conclude the proof by noting that
	\begin{align}
		\Pgivenop{X}{X'}{x}{x'}\Pgivenop{A'}{XAX'}{a'}{x,a,x'} = \sum_i\Pop{I}{i} \delta_{x,f_\xi(x')}\delta_{a',g'_i(a,x')}
		\,.
		\notag
	\end{align}
      \end{proof}

\newpage
\part{Invariant subspaces}
\label{Part:InvariantSubspaces}
We now motivate the analysis of the invariant subspaces of local transformations.
In particular, those invariant subspaces correspond to physical properties such as normalization, or being nonsignaling.
Indeed, local transformations map normalized behaviors to normalized behaviors.
As they act locally, they map nonsignaling behaviors to nonsignaling behaviors.
As normalization and nonsignaling constraints are written using linear equalities, they restrict the linear subspaces of which behaviors can be part of.
This second part of our manuscript is structured as follows.
Below, we provide an overview of the goals of this part using a concrete example, and show how single party invariant subspaces are linked to the multi-party structure. 
In Section~\ref{Sec:InvariantSubspacesSingle}, we construct the invariant subspaces of local transformations for the single party case.
We consider the multi-party case in Section~\ref{Sec:InvariantSubspacesMulti}, and identify the normalization, nonsignaling and signaling subspaces in arbitrary scenarios.
Applications are discussed in Section~\ref{Sec:InvariantApplications}: the equivalence of Bell-like inequalities, the optimization of their statistical properties, and the decomposition of assemblages/steering witnesses in steering scenarios.
Finally, Section~\ref{Sec:InvariantTechnical} proves the uniqueness of the decomposition into invariant subspaces, and collects longer proofs of the preceding sections.
This last section is more technical than the rest of the manuscript and should be skipped at first reading.

\section{Motivation}
\label{Section:InvariantSubspaces:Motivation}
\marginnote{
  To verify our understanding of that notation, we give a concrete example here.
  With the ordering of coefficients defined in Section~\ref{Sec:Definition:Behaviors}, we describe the PR box correlations with the vector
  \begin{multline*}
    \vecP{AB} = ( \Pnosub{11}{11}, \Pnosub{12}{11}, \\
    \Pnosub{11}{12}, \Pnosub{11}{22}, \\
    \ldots, \Pnosub{22}{22} )^\top
  \end{multline*}
  that is
  \begin{multline*}
    \vecP{AB}^\text{PR} = (1, 0, 1, 0, 0, 1, 0, 1, 1, \\
    0, 1, 0, 0, 1, 0, 1)^\top/2
  \end{multline*}
  while the CHSH inequality $\expr{AB}^\text{CHSH}$ has coefficients
  \begin{multline*}
    \expr{AB}^\text{CHSH} = (1, -1, 1, -1, -1, 1,\\
    -1, 1, 1, -1, -1, 1, -1, 1, 1, -1)
  \end{multline*}
  and we verify $\expr{AB}^\text{CHSH} ~ \vecP{AB}^\text{PR} = 4$.
  In the bases described above, these vectors decompose as
  \begin{multline*}
    \vecP{AB}^\text{PR} = \Puni{A} \otimes \Puni{B} + \\
                        \sum_{xy} (-1)^{(x-1)(y-1)} \data{A}^{x} \otimes \data{B}^{y}
  \end{multline*}
  while
  \begin{multline*}
    \expr{AB}^\text{CHSH} = \\
    \sum_{xy} (-1)^{(x-1)(y-1)} \datad{A}^{x} \otimes \datad{B}^{y}\;.
  \end{multline*}
}
Consider a two-party scenario with binary inputs and outputs without signaling between Alice and Bob.
The behavior vector $\vecP{AB}$ has 16 coefficients; however normalization constraints~\eqref{Eq:NormalizedSubset} and nonsignaling constraints (Definition~\ref{Def:Nonsignaling}) reduce the degrees of freedom needed to describe the behavior to 8.
Indeed, any nonsignaling behavior has a description using the correlators~\cite{Sliwa2003}
\begin{equation}
  \mathbf{A} = (-1)^{a-1}\;, \qquad \mathbf{B} = (-1)^{b-1}\;,
\end{equation}
with their expectation values compactly written
\begin{align}
  \left< \mathbf{A}_x \right> &= \sum_a (-1)^{a-1} \Pgivenop{A}{X}{a}{x}\;, \quad
                       \left< \mathbf{B}_y \right> = \sum_b (-1)^{b-1} \Pgivenop{B}{Y}{b}{y}\;, \nonumber \\
  \label{Eq:BinaryCorrelators}
  \left< \mathbf{A}_x \mathbf{B}_y \right> &= \sum_{a,b} (-1)^{(a-1)(b-1)} \Pgivenop{AB}{XY}{ab}{xy}\;.
\end{align}
for $a,b,x,y=1,2$;
The behavior $\vecP{AB}$ can be reconstructed with
\begin{equation}
  \vecP{AB} = \Puni{A} \otimes \Puni{B} + \sum_x \left< \mathbf{A}_x \right> ~ \data{A}^x \otimes \Puni{B} + \sum_y \left< \mathbf{B}_y \right> ~ \Puni{A} \otimes \data{B}^y + \sum_{xy} \left< \mathbf{A}_x \mathbf{B}_y \right> \data{A}^x \otimes \data{B}^y\;,
\end{equation}
  where
\begin{equation}
  \Puni{A} = \Puni{B} = \frac{1}{2} \begin{pmatrix} 1 \\ 1 \\ 1 \\ 1 \end{pmatrix}\;,
  \quad
  \data{A}^1 = \data{B}^1 = \frac{1}{2} \begin{pmatrix} 1 \\ -1 \\ 0 \\ 0 \end{pmatrix}\;,
  \quad
  \data{A}^2 = \data{B}^2 = \frac{1}{2} \begin{pmatrix} 0 \\ 0 \\ 1 \\ -1 \end{pmatrix}\;,
\end{equation}
using the coefficient enumeration convention of Section~\ref{Sec:Definition:Behaviors}.
The correlator notation provides several advantages: A behavior is specified using a minimal number of coefficients, Bell inequalities/expressions can be written in a canonical form under nonsignaling constraints, and it directly expresses the nature of correlations, for example whether marginals are uniformly random or not.

At the single party level ($\vecP{A} \in \RR^4$), we base our exploration of invariant subspaces on the following observation.
\begin{proposition}
  \label{Prop:LocalTransformationsNormalization}
  Let $\vecP{A}$ be an element of $\Pspace{A}$ such that $\sum_a \Pgivenop{A}{X}{a}{x} = c$, for all $x$ and a constant $c \in \RR$, but otherwise arbitrary.
  Let $\vecP{A}' = \Lmap{A} ~ \vecP{A}$ be the behavior after local transformation by an arbitrary $\Lmap{A}$.
  Then $\sum_{a'} \Pgivenopprime{A'}{X'}{a'}{x'} = c$.
\end{proposition}
\begin{proof}
  $\Lmap{A}$ is a mixture of deterministic transformations.
  For a single element in this decomposition:
  \begin{align}
    \sum_{a'} \Pgivenopprime{A'}{X'}{a'}{x'} & = \sum_{a'ax} \Pgivenop{A'}{AX'}{a'}{ax'} \Pgivenop{X}{X'}{x}{x'} \Pgivenop{A}{X}{a}{x} \\
    & = \sum_{x} \Pgivenop{X}{X'}{x}{x'} \underbrace{\sum_{a} \Pgivenop{A}{X}{a}{x}}_{= c} = c\;. \qedhere
  \end{align}
\end{proof}

An invariant subspace $\Vof{A} \subset \Pspace{A}$ is such that any $\vecP{A} \in \Vof{A}$ has image $\Lmap{A} ~ \vecP{A} \in \Vof{A}$.
Our observation is that there are nontrivial subspaces of $\Pspace{A}$ (read different from $0$ or $\Pspace{A}$) which are invariant.
A first invariant subspace is spanned by $\{\data{A}^1, \data{A}^2\}$, as $\sum_a \data{A}^x = 0$.

As local transformations preserve normalization itself (the value of the constant $c$ in Proposition~\ref{Prop:LocalTransformationsNormalization}), another invariant subspace is spanned by $\{\Puni{A}, \data{A}^1, \data{A}^2\}$.
The last invariant subspace is the space $\RR^4$ itself, for which we complete our basis with a fourth vector $\sig{A} = (1,1,-1,-1)^\top/2$, such that elements of the chain
\begin{equation}
  0 \subset \spanset{\data{A}^1, \data{A}^2} \subset \spanset{\data{A}^1, \data{A}^2, \Puni{A}} \subset \spanset{\data{A}^1, \data{A}^2, \Puni{A}, \sig{A}} = \RR^4 = \Pspace{A}
\end{equation}
are invariant subspaces of local transformations.
We will prove later that no other decomposition in terms of invariant subspaces is possible.

The basis $\left\{ \data{A}^1, \data{A}^2, \Puni{A}, \sig{A} \right\}$ induces a basis of the dual space $\Pspace{A}^*$, which we write $\left\{ \datad{A}^1, \datad{A}^2, \traceout{A}, \Sanh{A} \right\}$.
It is uniquely defined if we require that
\begin{equation}
  \datad{A}^x ~ \data{A}^x = \traceout{A} ~ \Puni{A} = \Sanh{A} ~ \sig{A} = 1
\end{equation}
and all other contractions equal to zero.

First, we have the map
\begin{equation}
  \label{Eq:InvariantSubspaces:Motivation:Ann}
  \Sanh{A}=(1,1,-1,-1)/2
\end{equation}
that verifies that the normalization of $\vecP{A}$ is balanced:
\begin{equation}
  \Sanh{A} ~ \vecP{A} = \sum_a \Pgivenop{A}{X}{a}{1} - \sum_a \Pgivenop{A}{X}{a}{2}\;.
\end{equation}
The subspace of $\Pspace{A}$ for which $\Sanh{A} ~ \vecP{A} = 0$ is exactly $\Vsub = \spanset{\data{A}^1, \data{A}^2, \Puni{A}}$.
Now, for any local transformation $\Lmap{A}$ and any $\vecP{A} \in \Vsub$, we have
\begin{equation}
  \Sanh{A} ~ (\Lmap{A} ~ \vecP{A}) = 0 \quad \Leftrightarrow \quad (\Sanh{A} ~ \Lmap{A}) ~ \vecP{A} = 0\;,
\end{equation}
and thus $(\Sanh{A} ~ \Lmap{A}) \in \spanset{\Sanh{A}}$, and $\spanset{\Sanh{A}}$ is an invariant subspace of $\Pspace{A}^*$.
As can be easily verified, another invariant subspace is given by $\spanset{\Sanh{A}, \traceout{A}}$, with $\traceout{A} = (1,1,1,1)/2$; we remark that $\Sanh{A} ~ \vecP{A} = 0$ and $\traceout{A} ~ \vecP{A} = 1$ are constraints obeyed by all normalized behaviors, and normalization is preserved by local transformations.
Finally, the dual basis is completed by $\datad{A}^1 = (1,-1,0,0)$ and $\datad{A}^2 = (0,0,1,-1)$.
Thus, the elements of the chain
\begin{equation}
  0 \subset \spanset{\Sanh{A}} \subset \spanset{\Sanh{A}, \traceout{A}} \subset \spanset{\Sanh{A}, \traceout{A}, \datad{A}^1, \datad{A}^2} = \Pspace{A}^*
\end{equation}
are invariant subspaces of $\Pspace{A}^*$.

We note that the value of the correlators~\eqref{Eq:BinaryCorrelators} can be recovered using these dual basis elements:
\begin{equation}
  \left<\mathbf{A}_x\right> = (\datad{A}^x \otimes \traceout{B}) ~ \vecP{AB}, \quad \left<\mathbf{B}_y\right> = (\traceout{A} \otimes \datad{B}^y) ~ \vecP{AB}, \quad \left<\mathbf{A}_x \mathbf{B}_y\right> = (\datad{A}^x \otimes \datad{B}^y) ~ \vecP{AB}\;.
\end{equation}
Finally, we remark that the above derivation was made for a two-party scenario without any signaling directions.
As we will see later, but can be verified explicitly by the reader, signaling from A to B is represented by the subspace $\spanset{\sig{A}} \otimes ~ \spanset{\data{B}^1, \data{B}^2}$, while signaling from B to A is represented by the subspace $\spanset{\data{A}^1, \data{A}^2} \otimes \spanset{\sig{B}}$.

As a conclusion of this overview, we see that a pertinent decomposition of the two-party behavior space comes directly from the decomposition of the invariant subspaces of single parties.

\clearpage

\section{Invariant subspaces of single party correlations}
\label{Sec:InvariantSubspacesSingle}

\marginnote{The constructions below can be generalized to arbitrary local transformations $\Lmap{A}: \Pspace{A} \to \Pspace{A'}$, where $\pst{A} \ne \pst{A'}$.
  The generalization is left to the reader, but the end result will be that the upper triangular form of Figure~\ref{Fig:LocalMapSpaces} applies to those maps as well.
}
We now study the invariant subspaces of the local transformation $\Lmap{A}: \Pspace{A} \to \Pspace{A}$, for a party A of arbitrary cardinality.
As $\Lmap{A}$ acts on behaviors and Bell expressions, we will study invariant subspaces of $\Pspace{A}$ and its dual $\Pspace{A}^*$.

\subsection{Linear forms}
\label{Sec:Def:LinearForms}
From Section~\ref{Sec:Def:BellExpressions}, we recall that a linear form $\expr{A}: \Pspace{A} \to \RR$ is an element of the dual space $\Pspace{A}^*$.
In the single party case, a useful example is the {\em trace out} map $\traceout{A} : \Pspace{A} \to \RR$, defined as:
\begin{equation}
  \label{Eq:Traceout}
  \traceout{A} ~ \vecP{A} = \frac{1}{X} \sum_{ax} \Pgivenop{A}{X}{a}{x} = \sum_{ax} \traceout{A}(a,x) ~ \Pgivenop{A}{X}{a}{x} \;, \qquad \traceout{A}(a,x) = 1/X\;.
\end{equation}
We observe that $\traceout{A} ~ \vecP{A} = 1$ for all normalized distributions.
Linear forms are written explicitly using their coefficients $\phi_\mathrm{A}(a,x)$ such that
\begin{equation}
  \expr{A} ~ \vecP{A} = \sum_{ax} \phi_\mathrm{A}(a,x) \Pgivenop{A}{X}{a}{x}\;.
\end{equation}

In the two-party case, linear forms respect the tensor structure, such that:
\begin{equation}
  (\expr{A} \otimes \expr{B}) \vecP{AB} = \sum_{abxy} \phi_\mathrm{A}(a,x) \phi_\mathrm{B}(b,y) \Pgivenop{AB}{XY}{ab}{xy}\;,
\end{equation}
and even partially applied
\begin{equation}
  \label{Eq:PartialApplication}
  \vecP{B}' = (\expr{A} \otimes \MidP{B}) \vecP{AB} \quad \Leftrightarrow \quad \Pgivenopprime{B}{Y}{b}{y} = \sum_{ax} \phi_\mathrm{A}(a,x) \Pgivenop{AB}{XY}{ab}{xy}\;.
\end{equation}

We observe that $(\traceout{A} \otimes \traceout{B}) ~ \vecP{AB} = 1$ for all normalized distributions, and that $(\traceout{A} \otimes \MidP{B}) ~ \vecP{AB}$ provides the marginal distribution $\vecP{B}$ for $x$ chosen uniformly at random.

\subsection{Notation for the Euclidean basis}
\label{Sec:InvariantSubspaces:Definitions}
\marginnote{If $\Pspace{A}$ was equipped with an inner product $\left<\cdot, \cdot \right>$, we could use the same basis for $\Pspace{A}$ and its dual, sending each basis element $\vec{e}_i$ to the dual element $\left<\vec{e}_i, \cdot \right>$.
  We do not have a meaningful inner product at hand, and as we will see, the spaces $\Pspace{A}$ and $\Pspace{A}^*$ have different decompositions into invariant subspaces.
  Thus, the usual approach based on orthogonal representations of finite or compact groups cannot work here.
  Our definitions for linear algebra are based on~\cite{Roman2005}, in particular Chapter 3.
  We summarize in this subsection the definitions required to understand the statements.
  The precise formulation of invariant subspaces and the proofs given in Section~\ref{Sec:InvariantTechnical} are based on the representation theory of associative algebras~\cite{Etingof2009}; however in the present section we keep the jargon to a minimum.
}
By identification $\Pspace{A} \sim \RR^{\dd{A}}$, we enumerate the elements of the standard basis on $x \in \domain{X}$ and $a\in\domain{A_x}$ (recall the definitions from Section~\ref{Sec:Definitions:Scenario}).
We denote by $\Pbasis{A}{a}{x}$ the basis element which has a single nonzero coefficient equal to one in the $(a,x)$ position, such that
\begin{equation}
  \boxed{
    \vecP{A} = \sum_{a,x} \Pgivenop{A}{X}{a}{x} \Pbasis{A}{a}{x}
  } \;.
\end{equation}
We enumerate $\Pbasis{A}{a}{x}^*$ the basis elements of $\Pspace{A}^*$, such that
\begin{equation}
  \boxed{
    \Pbasis{A}{a}{x}^* ~ \vecP{A} = \Pgivenop{A}{X}{a}{x}
  } \;.
\end{equation}
The bases $\left \{ \Pbasis{A}{a}{x} \right \}_{a,x}$ and $\left \{ \Pbasis{A}{a}{x}^* \right \}_{a,x}$ are dual (\cite[Theorem 3.11]{Roman2005}) in the sense that
\begin{equation}
  \Pbasis{A}{a'}{x'}^* ~ \Pbasis{A}{a}{x} = \delta_{a,a'} \delta_{x,x'}
\end{equation}
with $\delta$ the Kronecker delta.
This notation will be invaluable to define compactly a more appropriate basis of these spaces.

\subsection{Generalized correlators}
\label{Sec:InvariantSubspaces:Correlators}
\marginnote{
  Example: Let A have binary inputs and $A_1 = 3$ outputs for $x=1$ and $A_2 = 2$ outputs for $x=2$, i.e. $\pst{A} = (3,2)$.
  Then $\Puni{A} = (1/3,1/3,1/3,1/2,1/2)^\top$, $\vec{S}_1 = (1/3,1/3,1/3,-1/2,-1/2)^\top$, $\vec{C}_{1|1} = (1, 0, -1, 0, 0)^\top/3$, $\vec{C}_{2|1} = (0, 1, -1, 0, 0)^\top/3$, $\vec{C}_{1|2} = (0,0,0,1,-1)^\top/2$.
}
We now consider a basis of $\Pspace{A}$ and $\Pspace{A}^*$ that describes the physics at hand.
We first define the {\em uniformly random behavior}:
\begin{equation}
  \boxed{
    \Puni{A} = \sum_{x=1}^X \sum_{a=1}^{A_x} \frac{1}{A_x} \Pbasis{A}{a}{x}
  }\;,
\end{equation}
the {\em correlation vectors}:
\begin{equation}
  \label{Eq:CorrelationVectors}
  \boxed{
    \data{A}^{i|x} = \frac{\Pbasis{A}{i}{x} - \Pbasis{A}{A_x}{x}}{A_x}
  }\;, \quad x \in \domain{X}, \quad i \in \domain{A_x-1},
\end{equation}
\marginnote{Motivating a unique form for those definitions is one goal of the present manuscript.
  Part of it comes from the existence of invariant subspaces, while other choices come from existing conventions used in the field (Section~\ref{Sec:InvariantTechnical:FixingDOF}).
}
and the {\em normalization-violating} or {\em signaling vectors}:
\begin{equation}
  \boxed{
    \sig{A}^k =  \sum_{a=1}^{A_k} \frac{X}{A_k} \Pbasis{A}{a}{k} - \sum_{x=1}^X \sum_{a=1}^{A_x} \frac{1}{A_x} \Pbasis{A}{a}{x}
  }\;, \quad k\in\domain{X-1}.
\end{equation}
We write
\begin{equation}
  \boxed{
    \Zof{A} = \spanset{\Puni{A}}, \quad \Cof{A} = \spanset{\left\{\data{A}^{i|x}\right\}_{ix}}, \quad \Sof{A} = \spanset{\left\{\sig{A}^k\right\}_k}
  }
\end{equation}
the subspaces spanned by those vectors.

\subsubsection{Basis of the dual space}
\marginnote{
  Example: let $\pst{A} = (3,2)$ as above.
  Then $\traceout{A} = (1,1,1,1,1)/2$, $\Sanh{A}^1 = (1,1,1,-1,-1)/2$, $\datad{A}^{1|1} = (2,-1,-1,0,0)$, $\datad{A}^{2|1} = (-1,2,-1,0,0)$ and $\datad{A}^{1|2} = (0,0,0,1,-1)$.
}
Now, we move to the space $\Pspace{A}^*$ of Bell expressions.
We define the normalization checking linear forms $\Ssum{A}^x \in \Pspace{A}^*$:
\begin{equation}
  \label{Eq:DefSsum}
  x\in\domain{X}, \quad \Ssum{A}^x = \sum_{a=1}^{A_x} \Pbasis{A}{a}{x}^* \;,
\end{equation}
so that $\Ssum{A}^x \vecP{A} = 1$ for normalized distributions. 
From those, we can express the {\em trace out form} of Eq.~\eqref{Eq:Traceout} that verifies overall normalization ($\traceout{A} \vecP{A} = 1$) and discards parties:
\begin{equation}
  \label{Eq:Traceout1}
  \boxed{ \traceout{A} = \frac{1}{X} \sum_{x=1}^X \Ssum{A}^x } \;.
\end{equation}
Uniform normalization is checked by the forms $\Sanh{A}^k$:
\begin{equation}
  \label{Eq:DefSanh}
  \boxed{
    \Sanh{A}^k = \frac{\Ssum{A}^k - \Ssum{A}^X}{X}
  } \;, \quad k\in\domain{X-1},
\end{equation}
so that $\Sanh{A}^k~\vecP{A} = 0$ for normalized distributions.
We complete our basis by the forms $\datad{A}^{i|x}$:
\marginnote{Note that $\datad{A}^{i|x}$ corresponds to the usual binary correlators when $A_x=2$, which were discussed in the introductory Section~\ref{Section:InvariantSubspaces:Motivation}.
  Writing $\left< \mathbf{A}_x \right> = \Pgivenop{A}{X}{1}{x} - \Pgivenop{A}{X}{2}{x}$, we have, for normalized $\vecP{A}$: $\Sanh{A}^1 ~ \vecP{A} = 0$, $\traceout{A} ~ \vecP{A} = 1$ and $\datad{A}^{1|x} ~ \vecP{A} = \left< \mathbf{A}_x \right>$.}
\begin{equation}
  \label{Eq:Defdatad}
  \boxed{
    \datad{A}^{i|x} ~ \vecP{A} = A_x ~ \Pbasis{A}{i}{x}^* - \sum_{a} \Pbasis{A}{a}{x}^*
    } \;, \quad x\in\domain{X}, \quad i\in\domain{A_x-1},
\end{equation}

We now identify subspaces of $\Pspace{A}^*$:
\begin{equation}
  \boxed{
    \Zof{A}^* = \spanset{\traceout{A}}, \quad \Sof{A}^* = \spanset{\left\{\Sanh{A}^i\right\}_i}, \quad \Cof{A}^* = \spanset{\left\{\datad{A}^{i|x}\right\}_{ix}}
  }
\end{equation}
where the subspace $\Zof{A}^*$ corresponds to linear maps that take a constant value on normalized distributions, while $\Sof{A}^*$ corresponds to normalization-checking forms; finally $\Cof{A}^*$ corresponds to the forms that extract the correlation data from behaviors.

\subsection{Basis duality and projection on subspaces}
The bases given above are dual to each other.

\begin{proposition}
  \label{Prop:DualBasisNonZeros}
  Let $\left\{ \Puni{A} \right\} \cup \left\{ \data{A}^{i|x} \right \} \cup \left \{ \sig{A}^k \right \}$ be a basis of $\Pspace{A}$ for $x \in \domain{X}$, $k \in \domain{X-1}$ and $i \in \domain{A_x - 1}$.
  Let $\left\{ \traceout{A} \right\} \cup \left\{ \datad{A}^{i'|x'} \right\} \cup \left\{ \Sanh{A}^{k'} \right\}$ also be an basis of $\Pspace{A}^*$, for $x' \in \domain{X}$, $k' \in \domain{X-1}$ and $i' \in \domain{A_x - 1}$.
  Those two bases are dual to each other.
  In particular, the only nonzero contractions are
  \begin{equation}
    \traceout{A} ~ \Puni{A} = 1, \quad \datad{A}^{i'|x'} ~ \data{A}^{i|x} = \delta_{i',i} \delta_{x',x}, \quad \Sanh{A}^{k'} ~ \sig{A}^k = \delta_{i',i}\;,
  \end{equation}
  as shown in Table~\ref{Tab:Subspaces}.
\end{proposition}
\begin{proof}
  Left to the reader.
\end{proof}
This means that for $\expr{A} \in \Cof{A}^*$ or $\Zof{A}^*$ or $\Sof{A}^*$, and $\vecP{A} \in \Cof{A}$ or $\Zof{A}$ or $\Sof{A}$, the product $\expr{A} ~ \vecP{A}$ can be nonzero only for pairs of vectors in $(\Cof{A}^*,\Cof{A})$, $(\Zof{A}^*,\Zof{A})$ and $(\Sof{A}^*,\Sof{A})$.

Keeping the interpretation of behaviors as column vectors and linear forms as row vectors, we obtain the following corollary.
\begin{corollary}
  \label{Cor:Projectors}
  The projectors on the subspaces $\Zof{A}$, $\Cof{A}$ and $\Sof{A}$ are written:
  \begin{equation}
    \Pi_{\Zof{A}} = \Puni{A} \cdot \traceout{A}, \quad \Pi_{\Cof{A}} = \sum_{ix} \data{A}^{i|x} \cdot \datad{A}^{i|x}, \quad \Pi_{\Sof{A}} = \sum_k \sig{A}^k \cdot \Sanh{A}^k\;.
  \end{equation}
\end{corollary}
And as we see below, the projectors $\Pi_{\Zof{A}}$ and $\Pi_{\Zof{A}} + \Pi_{\Cof{A}}$ are uniquely defined by invariance under local transformations.
\begin{table}
  \begin{center}
    \begin{tabular}{r|c|c|c}
      & $\data{A}^{i|x} \in \Cof{A}$ & $\Puni{A} \in \Zof{A}$ & $\sig{A}^k \in \Sof{A}$ \\
      \hline
      $\datad{A}^{i'|x'} \in \Cof{A}^*$ & $\delta_{i',i} \delta_{x',x}$ & $0$ & $0$ \\
      \hline
      $\traceout{A} \in \Zof{A}^*$ & $0$ & $1$ & $0$ \\
      \hline
      $\Sanh{A}^{k'} \in \Sof{A}^*$ & $0$ & $0$ & $\delta_{k',k}$
    \end{tabular}
  \end{center}
  \caption{
    \label{Tab:Subspaces} Relations between linear forms and subspaces when computing $\phi(\vec{v})$ for $\vec{v} \in V$ where $\phi$ is a linear form and $V$ a subspace.
  }
\end{table}

\subsection{Decomposition of invariant subspaces}
\label{Sec:InvariantSubspaces:Decomposition}
The bases above are motivated by the following decomposition.
\marginnote{
 \begin{center}
     \includegraphics{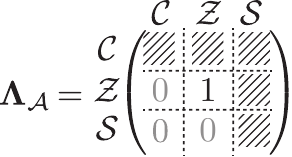}
   \end{center}
   \captionof{figure}{\label{Fig:LocalMapSpaces}
     Block diagonal form of a local map $\Lmap{A}$ in the basis given by the subspaces $\Csub$, $\Zsub$ and $\Ssub$.
   }
 }
\begin{proposition}[Simple version]
  \label{Prop:Decomposition}
  The space $\Pspace{A}$ decomposes as the series of subspaces invariant under local transformations
  \begin{equation}
    0 \subset \Cof{A} \subset \Cof{A} \oplus \Zof{A} \subset \Cof{A} \oplus \Zof{A} \oplus \Sof{A} = \Pspace{A}
  \end{equation}
  and the decomposition is unique.
\end{proposition}
\begin{proof}
  See Section~\ref{Sec:InvariantTechnical}, where it is reformulated as Proposition~\ref{Prop:Filtration}.
\end{proof}
Note that the proposition allows freedom in the definition of the subspace $\Zof{A}$; only the direct sum $\Cof{A} \oplus \Zof{A}$ is unique, not the factor $\Zof{A}$ itself.
We show in Section~\ref{Sec:InvariantTechnical:RelabelingInvariant} that $\Zof{A}$ and $\Sof{A}$ are uniquely determined if we ask additionally that both subspaces are invariant under relablings of outputs, and that $\traceout{A}$ averages uniformly over inputs.
Proposition~\ref{Prop:Decomposition} can be reformulated by saying that the local transformation $\Lmap{A}$
has a block triangular form as in Figure~\ref{Fig:LocalMapSpaces}.
\marginnote{This block triangular form extends to maps $\Lmap{A}: \Pspace{A} \to \Pspace{A'}$ between spaces of different cardinalities.}
This decomposition induces a decomposition of the dual space $\Pspace{A}^*$.
We observe that $\Sof{A}^*$ is zero for elements of $\Cof{A} \oplus \Zof{A}$, and that $\Sof{A}^* \oplus \Zof{A}^*$ is zero for elements of $\Cof{A}$; Lemma~\ref{Lemma:AnnihilatorDualDecomposition} in Section~\ref{Sec:InvariantTechnical} provides the following corollary.
\begin{corollary}
  \label{Prop:DecompositionDual}
  The space $\Pspace{A}^*$ decomposes as the series of subspaces invariant under local transformations
  \begin{equation}
    0 \subset \Sof{A}^* \subset \Sof{A}^* \oplus \Zof{A}^* \subset \Sof{A}^* \oplus \Zof{A}^* \oplus \Cof{A}^* = \Pspace{A}^*
  \end{equation}
  and the decomposition is unique.
\end{corollary}

We now discuss two impacts of this decomposition at the single party level.

\subsection{Impact on single party behaviors}
\label{Sec:SplitNormalizationSingle}
Using the linear forms defined above, the normalization constraint is equivalently written:
\begin{equation}
  \sum_{a=1}^{A_x} \Pgivenop{A}{X}{a}{x} = 1 \quad \text{or} \quad \Ssum{A}^x ~ \vecP{A} = 1\;, \qquad x\in\domain{X}
\end{equation}
which we split into
\begin{equation}
  \boxed{\traceout{A} ~ \vecP{A} = 1} \qquad \text{and} \qquad \boxed{\Sanh{A}^k ~ \vecP{A} = 0}, \quad k\in\domain{X-1} \;.
\end{equation}

Using the dual basis relations (Proposition~\ref{Prop:DualBasisNonZeros}), we see that any normalized $\vecP{A}$ can be written as
\begin{equation}
  \boxed{
    \vecP{A} = \Puni{A} + \vec{C}, \qquad \vec{C} \in \Cof{A}
  }\;,
\end{equation}
as $\vecP{A}$ cannot have support in $\Sof{A}$ (because $\Sanh{A}^i ~ \vecP{A} = 0$), and its coefficient in $\Zof{A}$ is fixed by $\traceout{A} ~ \vecP{A} = 1$.

\subsection{The Collins-Gisin basis}
\marginpar{
  These definitions apply directly to the multi-party case by performing the change of basis on each component of the tensor product space.
  For example, in the CHSH scenario, the Collins-Gisin basis for Alice is given by $(1, \Pgivenop{A}{X}{1}{1}, \Pgivenop{A}{X}{1}{2})^\top$ and the Collins-Gisin basis for Bob is $(1, \Pgivenop{B}{Y}{1}{1}, \Pgivenop{B}{Y}{1}{2})^\top$.
  The resulting tensor product has $9$ elements: one constant normalization, $4$ single party marginals and $4$ two-party coefficients.
}
Another description of the nonsignaling space is given by the Collins-Gisin basis.
Given a probability distribution described in that basis, the distribution in the full space is readily reconstructed.
For example, for $\pst{A} = (3,2)$,
\begin{equation}
  \begin{pmatrix} \Pgivenop{A}{X}{1}{1} \\ \Pgivenop{A}{X}{2}{1} \\ \Pgivenop{A}{X}{3}{1} \\ \Pgivenop{A}{X}{1}{2} \\ \Pgivenop{A}{X}{2}{2} \end{pmatrix} = \underbrace{\begin{pmatrix} 0 & 1 & 0 & 0 \\ 0 & 0 & 1 & 0 \\ 1 & -1 & -1 & 0 \\ 0 & 0 & 0 & 1 \\ 1 & 0 & 0 & -1 \end{pmatrix}}_{G_\Pspace{A}} \begin{pmatrix} 1 \\ \Pgivenop{A}{X}{1}{1} \\ \Pgivenop{A}{X}{2}{1} \\ \Pgivenop{A}{X}{1}{2} \end{pmatrix}\;.
\end{equation}

As $G_\Pspace{A}$ is not square, its inverse is not uniquely defined.
We can resolve this by prescribing the following.
We write $G_\Pspace{A}^+$ for the matrix that satisfies $G_\Pspace{A}^+ ~ G_\Pspace{A} = \Mid$ (and thus is a pseudoinverse), and require that $G_\Pspace{A} ~ G_\Pspace{A}^+$ is a projector on the $\Cof{A} \oplus \Zof{A}$ subspace.
The inverse coresponding to the example above is
\begin{equation}
  G_\Pspace{A}^+ = \frac{1}{12} \begin{pmatrix} 6 & 6 & 6 & 6 & 6 \\ 10 & -2 & -2 & 2 & 2 \\ -2 & 10 & -2 & 2 & 2 \\ 3 & 3 & 3 & 9 & -3 \end{pmatrix}\;.
\end{equation}

In the general case, the matrix $G_\Pspace{A}$ has the following form:
\begin{equation}
  G_\Pspace{A} = \begin{pmatrix}
    \vec{f}_1 & K_1 & 0 & 0 & \ldots \\
    \vec{f}_2 & 0 & K_2 & 0 & \\
    \vec{f}_A & 0 & 0 & K_3 & \\
    \vdots & & &
  \end{pmatrix}, \qquad \vec{f}_x = \begin{pmatrix} \vec{0}_{\alpha_x} \\ 1 \end{pmatrix}, \qquad K_n = \begin{pmatrix} \Mid_{\alpha_x} \\ -\vec{1}_{\alpha_x}^\top \end{pmatrix}\;,
\end{equation}
where $\alpha_x = A_x - 1$, while $\vec{0}_n \in \RR^n$ is the vector of all zeros and $\vec{1}_n \in \RR^n$ is the vector of all ones.
\begin{proposition}
  \label{Prop:CollinsGisinInverse}
  The following matrix $G_\Pspace{A}^+$ satisfies $G_\Pspace{A} ~ G_\Pspace{A}^+ = \Mid$ so that $G_\Pspace{A}^+ ~ G_\Pspace{A}$ is a projector on the $\Cof{A} \oplus \Zof{A}$ subspace.
  \begin{equation}
    G_\Pspace{A}^+ = \begin{pmatrix}
      \vec{1}_{A_1}^\top/X & \vec{1}_{A_2}^\top/X & \vec{1}_{A_3}^\top/X & \ldots \\
      H_{11} & H_{12} & H_{13} & \\
      H_{21} & H_{22} & H_{23} & \\
      H_{31} & H_{32} & H_{33} & \\
      \vdots & & &
    \end{pmatrix}, \qquad
    \begin{array}{rl}
      H_{xx} & = \begin{pmatrix} \Mid_{\alpha_x} \!\!- \! \mu_x \boldsymbol{1}_{\alpha_x \times \alpha_x} & -\mu_x \vec{1}_{\alpha_x} \end{pmatrix}\;, \\
      H_{i\ne j} & = \nu_x \boldsymbol{1}_{\alpha_i \times A_j} \;, \\
      \mu_x & = (n-1)/(X A_x) \;, \\
      \nu_x & = 1/(X A_x) \;,
    \end{array}
  \end{equation}
    where $\boldsymbol{1}_{m \times n} \in \RR^{m \times n}$ is the matrix of all ones.
\end{proposition}
\begin{proof}
  Left to the reader (straight forward calculation).
\end{proof}

\clearpage

\section{Invariant subspaces of multi-party correlations}
\label{Sec:InvariantSubspacesMulti}
We now consider the invariant subspaces of multi-party correlations.
In particular, we link the multi-party normalization and nonsignaling constraints (Definition~\ref{Def:Nonsignaling}) and the decomposition defined in Section~\ref{Sec:InvariantSubspacesSingle}.
We study first the two-party
case and provide explicit characterizations; we then discuss multi-party generalizations.
\subsection{Two-party distributions}

We consider the space $\Pspace{A} \otimes \Pspace{B}$ of two-party behaviors and its dual $(\Pspace{A} \otimes \Pspace{B})^*$ containing two-party Bell expressions.
Due to the tensor structure, for any invariant subspace
\begin{equation}
  \Vof{A} \quad = \quad \Cof{A},\quad \Cof{A} \oplus \Zof{A} \quad \text{ or } \quad \Cof{A} \oplus \Zof{A} \oplus \Sof{A}\;,
\end{equation}
(and the same for $\Vof{B}$), the subspace $\Vof{A} \otimes \Vof{B}$ is invariant under $\Lmap{A} \otimes \Lmap{B}$ for arbitrary local transformations $\Lmap{A}$ and $\Lmap{B}$.
Our goal is now to provide an interpretation for the nine combinations $\Cof{A} \otimes \Cof{B}$, $\Cof{A} \otimes (\Cof{B} \oplus \Zof{B})$ and so on.
\begin{proposition}
  \label{Prop:NonDegenerateConstraints}
  The behavior $\vecP{AB}$ is normalized if and only if (iff.) it satisfies the following constraints:
  \begin{equation}
    \boxed{
      (\traceout{A} \otimes \traceout{B}) ~ \vecP{AB} = 1 \quad \text{and} \quad (\Sanh{A}^k \otimes \traceout{B}) ~ \vecP{AB} = (\traceout{A} \otimes \Sanh{B}^l) ~ \vecP{AB} = (\Sanh{A}^k \otimes \Sanh{B}^l) ~ \vecP{AB} = 0
    }\;.
  \end{equation}
  In addition, a normalized~$\vecP{AB}$ is nonsignaling from A to B iff. it satisfies
  \begin{equation}
    \boxed{
      (\Sanh{A}^k \otimes \datad{B}^{j|y}) ~ \vecP{AB} = 0
    }\;,
  \end{equation}
  and nonsignaling from B to A iff. it satisfies
  \begin{equation}
    \boxed{
      (\datad{A}^{i|x} \otimes \Sanh{B}^l) ~ \vecP{AB} = 0
    }\;.
  \end{equation}
  In these definitions, the indices $i,j,k,l,x,y$ run over their respective domains.
\end{proposition}
\begin{proof}
  Remark that
\begin{equation}
  (\Ssum{A}^x \otimes \Ssum{B}^y) ~ \vecP{AB} = 1, \quad \forall x,y\;,
\end{equation}
simply expresses the normalization constraint $\sum_{ab} \Pgivenop{AB}{XY}{ab}{xy} = 1$.
As in Section~\ref{Sec:SplitNormalizationSingle}, we split that constraint into
\begin{equation}
  \label{Eq:TwoParty:Normalization}
  (\traceout{A} \otimes \traceout{B}) ~ \vecP{AB} = 1, \qquad (\traceout{A} \otimes \Sanh{B}^l) ~ \vecP{AB} = (\Sanh{A}^k \otimes \traceout{B}) ~ \vecP{AB} = (\Sanh{A}^k \otimes \Sanh{B}^l) ~ \vecP{AB} = 0
\end{equation}
for all $i,j$.
This provides an interpretation for the four subspaces $(\Zof{A} \oplus \Sof{A}) \otimes (\Zof{B} \oplus \Sof{B})$.
We consider now the constraint that A does not signal to B.
It is written
\begin{equation}
  \sum_a \Pgivenop{AB}{XY}{ab}{xy} - \Pgivenop{AB}{XY}{ab}{x'y} = 0 \quad \Leftrightarrow \quad \left[ (\Ssum{A}^x - \Ssum{A}^{x'}) \otimes \MidP{B}\right] \vecP{AB} = 0
\end{equation}
for all $b,x,x',y$.
Without loss of generality, we can fix $x'=X$ to the last input value.
Then the nonsignaling constraint becomes:
\begin{equation}
  (\Sanh{A}^k \otimes \traceout{B}) ~ \vecP{AB} = (\Sanh{A}^k \otimes \datad{B}^{j|y}) ~ \vecP{AB} = (\Sanh{A}^k \otimes \Sanh{B}^l) ~ \vecP{AB} = 0\;.
\end{equation}
Compared to the normalization constraint~\eqref{Eq:TwoParty:Normalization}, only $(\Sanh{A}^k \otimes \datad{B}^{j|y}) ~ \vecP{AB} = 0$ is new, which leads us to identify $\Sof{A} \otimes \Cof{B}$ as the A $\to$ B signaling subspace.
A similar argument shows that $\Cof{A} \otimes \Sof{B}$ corresponds to the B $\to$ A signaling subspace, and that $(\datad{A}^{i|x} \otimes \Sanh{B}^l) ~ \vecP{AB} = 0$ is the B $\to$ A nonsignaling constraint not covered by normalization.
\end{proof}
Regarding the correlation space $\Pspace{A} \otimes \Pspace{B}$, we obtain the following characterization.

\begin{proposition}
  \label{Prop:TwoPartyBehaviorDecomposition}
  Any behavior $\vecP{AB}$ has the form
\begin{equation}
  \vecP{AB} = \Puni{A} \otimes \Puni{B} + \vecP{AB}^\text{nonsig} + \vecP{AB}^{\text{A}\to\text{B}} + \vecP{AB}^{\text{B}\to\text{A}}
\end{equation}
  where the nonsignaling component $\vecP{AB}^\text{nonsig}$ is in $(\Zof{A} \otimes \Cof{B}) \oplus (\Cof{A} \otimes \Zof{B}) \oplus (\Cof{A} \otimes \Cof{B})$, signaling A$\to$B is expressed by $\vecP{AB}^{\text{A}\to\text{B}} \in \Sof{A} \otimes \Cof{B}$, and signaling B$\to$A is expressed by $\vecP{AB}^{\text{B}\to\text{A}} \in \Cof{A} \otimes \Sof{B}$.
A graphical summary is displayed in Figure~\ref{Fig:TwoPartySubspaces}.
\end{proposition}
\marginnote{
  \begin{center}
	  \includegraphics{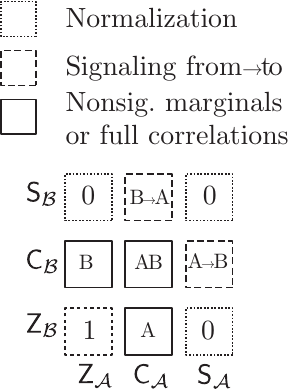}
  \end{center}
  \captionof{figure}{\label{Fig:TwoPartySubspaces}
    Subspaces in two party (A,B) scenarios.
  }
}[-4cm]
\marginnote{Signaling directions are allowed or forbidden depending on the scenario definition: see Section~\ref{Sec:Definitions:Scenario}.}[3cm]
\begin{proof}
  Looking back at Proposition~\ref{Prop:NonDegenerateConstraints}, normalization fixes the component in $\Zof{A} \otimes \Zof{B}$, and forbids components in $\Zof{A} \otimes \Sof{B}$, $\Sof{A} \otimes \Zof{B}$ and $\Sof{A} \otimes \Sof{B}$.
  The signaling components $\Cof{A} \otimes \Sof{B}$ and $\Sof{A} \otimes \Cof{B}$ were identified from the signaling constraints.
  Remain the interpretation of $\Zof{A} \otimes \Cof{B}$, $\Cof{A} \otimes \Zof{B}$ and $\Cof{A} \otimes \Cof{B}$.
  Remark that elements in $\Cof{A} \otimes \Cof{B}$ are sent to zero when tracing out one of the parties, either using $(\traceout{A} \otimes \MidP{B})$ or $(\MidP{A} \otimes \traceout{B})$.
  Thus, $\Cof{A} \otimes \Cof{B}$ corresponds to joint correlations.
  On the other hand, $\Cof{A} \otimes \Zof{B}$ disappears when tracing out B, and thus corresponds to the marginal correlations of A; similarly for $\Zof{A} \otimes \Cof{B}$, which disappears when tracing out A, and represents the marginal correlations of B.
\end{proof}

\subsection{Normalization and (non)signaling subspaces in multi-party scenarios}
\label{Sec:InvariantSubspaces:MultipartySubspaces}
We now move to the multi-party case.
For $n$ parties, the space $\Pspace{A} \otimes \Pspace{B} \otimes \Pspace{C} \otimes \ldots$ decomposes as
\begin{equation}
  \Pspace{A} \otimes \Pspace{B} \otimes \Pspace{C} \otimes \ldots =
  \left( \Sof{A} \oplus \Cof{A} \oplus \Zof{A} \right) \otimes
  \left( \Sof{B} \oplus \Cof{B} \oplus \Zof{B} \right) \otimes
  \left( \Sof{C} \oplus \Cof{C} \oplus \Zof{C} \right) \otimes \ldots
\end{equation}

We consider a single term $\Vof{A} \otimes \Vof{B} \otimes \Vof{C} \otimes \ldots$ in the expansion of the decomposition above, with $\Vof{A} \in \left \{ \Sof{A}, \Cof{A}, \Zof{A} \right \}$, $\Vof{B} \in \left \{ \Sof{B}, \Cof{B}, \Zof{B} \right \}$, and so on, so that after expansion of the tensor products we are left with $3^n$ subspaces.
The question is now to identify what a given subspace corresponds to.
We count using $n_\Ssub$, $n_\Csub$ and $n_\Zsub$ how many times each subspace is present, with $n_\Ssub + n_\Csub + n_\Zsub = n$.
To ease the notation in the propositions below, we reorder the parties such that $A_1, \ldots, A_{n_\Zsub}$ correspond to subspaces of type $\Zsub$; that $B_1, \ldots, B_{n_\Ssub}$ correspond to subspaces of type $\Ssub$, and finally $C_1, \ldots, C_{n_\Csub}$ correspond to subspaces of type $\Csub$.

We are now ready to generalize the Propositions~\ref{Prop:NonDegenerateConstraints} and~\ref{Prop:TwoPartyBehaviorDecomposition} to the multi-party case.
We will consider the different combinations of subspaces separately.

\marginnote{
  \begin{center}
    \includegraphics{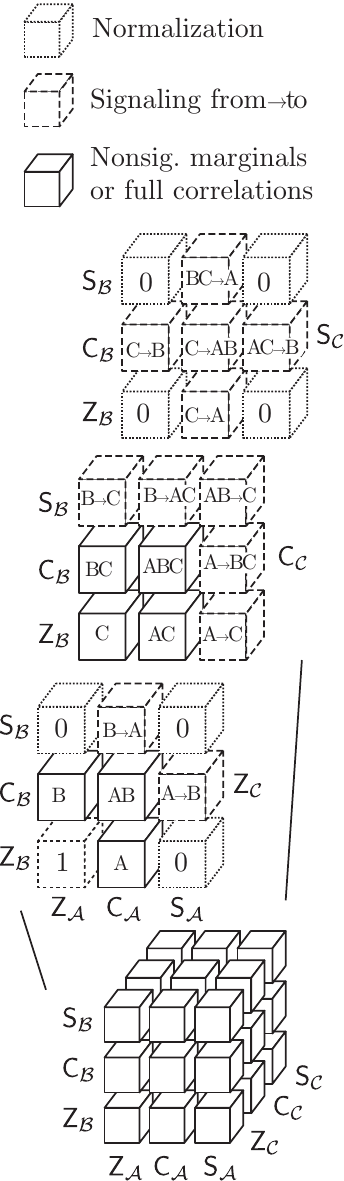}
  \end{center}
  \captionof{figure}{\label{Fig:ThreePartySubspaces}
    Subspaces in three party (A,B,C) scenarios; note that we have ordered the components as $\Zsub$, $\Csub$ and $\Ssub$ for clarity.
  }
}[-6cm]
\begin{proposition}[Normalization]
  \label{Prop:Normalization}
  The subspaces with $n_\Csub = 0$ correspond to {\em normalization subspaces}.
  For $n_\Zsub = n$ and $n_\Ssub = 0$, we have the constraint
  \begin{equation}
    (\traceoutI{A}{1} \otimes \traceoutI{A}{2} \otimes \ldots) ~ \vec{P}_{\mathcal{A}_1\mathcal{A}_2\ldots} = 1\;,
  \end{equation}
  and thus, for $\vecP{}$, the component in the $\Zsub_{\mathcal{A}_1} \otimes \Zsub_{\mathcal{A}_2} \otimes \ldots$ subspace is fixed to $\PuniI{A}{1} \otimes \PuniI{A}{2} \otimes \ldots$.
  For $n_\Ssub = n$ and $n_\Zsub = 0$, we have the constraint
  \begin{equation}
    (\SanhI{B}{1}^{k_1} \otimes \ldots \otimes \SanhI{B}{n_\Ssub}^{k_{n_\Ssub}}) ~ \vec{P}_{\mathcal{A}_1 \ldots \mathcal{B}_{n_\Ssub}} = 0\;,
  \end{equation}
  while when $n_\Zsub, n_\Ssub > 0$, we obtain
  \begin{equation}
    (\traceoutI{A}{1} \otimes \ldots \otimes \traceoutI{A}{n_\Zsub} \otimes \SanhI{B}{1}^{k_1} \otimes \ldots \otimes \SanhI{B}{n_\Ssub}^{k_{n_\Ssub}}) ~ \vec{P}_{\mathcal{A}_1 \ldots \mathcal{B}_{n_\Ssub}} = 0\;,
  \end{equation}
  for all $k_1,\ldots,k_{n_\Ssub}$,
  and there cannot be a component in the $\Zsub_{\mathcal{A}_1} \otimes \ldots \otimes \Zsub_{\mathcal{A}_{n_\Zsub}} \otimes \Ssub_{\mathcal{B}_1} \otimes \ldots \otimes \Ssub_{\mathcal{B}_{n_\Ssub}}$ subspace for any normalized behavior $\vecP{}$.
\end{proposition}
\begin{proof}
  See Section~\ref{Sec:ProofPropNormalization}.
\end{proof}

The subspaces with $n_\Csub > 0$ and $n_\Ssub = 0$ correspond to {\em nonsignaling correlations}.
\begin{proposition}[Nonsignaling correlations]
  \label{Prop:Nonsignaling}
  Any deterministic nonsignaling behavior has nonzero support in the subspace $\Zsub_{\mathcal{A}_1} \otimes \ldots \otimes \Zsub_{\mathcal{A}_{n_\Zsub}} \otimes \Csub_{\mathcal{C}_1} \otimes \ldots \otimes \Csub_{\mathcal{C}_{n_\Csub}}$.
\end{proposition}
\begin{proof}
  See Section~\ref{Sec:ProofPropNonsignaling}.
\end{proof}
With our definition of the trace out map~\eqref{Eq:Traceout} and the duality relations of Proposition~\ref{Prop:DualBasisNonZeros}, we remark that the subspaces of $\mathcal{B}_1 \otimes \ldots \otimes \mathcal{B}_{n_\Ssub} \otimes \mathcal{C}_1 \otimes \mathcal{C}_{n_\Csub}$ express the marginal distribution $\Pgivenop{B_1B_2\ldots C_1C_2\ldots}{Y_1Y_2\ldots Z_1Z_2 \ldots}{b_1b_2\ldots c_1c_2\ldots}{y_1y_2\ldots z_1z_2\ldots}$ after the parties $\mathcal{A}_1$, $\mathcal{A}_2$ have been traced out using a uniform input distribution $\operatorname{P}(x_1x_2\ldots) = 1/(X_1 X_2 \ldots)$.
\begin{proposition}[Signaling correlations]
  \label{Prop:Signaling}
  Let $n_\Ssub > 0$ and $n_\Csub > 0$.
  If $(B_i, C_j) \notin E$ for all $i,j$, i.e. no Bob signals to any Charlie, then the behaviors of that scenario obey the constraint:
  \begin{equation}
    (\traceoutI{A}{1} \otimes \traceoutI{A}{2} \otimes \ldots \otimes \SanhI{B}{1}^{l_1} \otimes \SanhI{B}{2}^{l_2} \otimes \ldots \otimes \datadI{C}{1}^{k_1|x_1} \otimes \datadI{C}{2}^{k_2|x_2}) ~ \vec{P} = 0\;,
  \end{equation}
  and thus do not have a component in the subspace
\begin{equation}
    \Zsub_{\mathcal{A}_1} \otimes \ldots \otimes \Zsub_{\mathcal{A}_{n_\Zsub}} \otimes \Ssub_{\mathcal{B}_1} \otimes \ldots \otimes \Ssub_{\mathcal{B}_{n_\Ssub}} \otimes \Csub_{\mathcal{C}_1} \otimes \ldots \otimes \Csub_{\mathcal{C}_{n_\Csub}}\;.
  \end{equation}
\end{proposition}
\begin{proof}
  See Section~\ref{Sec:ProofPropSignaling}.
\end{proof}

A graphical summary is displayed in Figure~\ref{Fig:ThreePartySubspaces} for the three party case.

\clearpage

\section{Applications}
\label{Sec:InvariantApplications}
We now discuss a few applications.
We discuss in the first part (Section~\ref{Sec:InvariantApplications:Equivalence}) the equivalence of Bell-like inequalities due to the normalization and nonsignaling constraints.
In the second part (Section~\ref{Sec:InvariantApplications:Variance}) we generalize the method proposed in~\cite{Renou2016} to optimize the variance of Bell inequalities used as statistical estimators.
Finally, we present a decomposition of assemblages/witnesses applicable to steering scenarios (Section~\ref{Sec:InvariantApplications:Steering}).
\subsection{Equivalence of inequalities}
\label{Sec:InvariantApplications:Equivalence}
\marginnote{For more parties: define $\traceout{AB\ldots} = \traceout{A} \otimes \traceout{B} \otimes \ldots$, and select the $\{ \mu_i \}$ as the linear forms corresponding to subspaces that are not present (Propositions~\ref{Prop:Normalization} and~\ref{Prop:Signaling}).}
The equivalence of inequalities under affine constraints was studied in~\cite{Rosset2014a} for nonsignaling scenarios.
The construction presented in that paper generalizes readily to scenarios involving signaling directions, as we show now.
We group the linear constraints present in Proposition~\ref{Prop:NonDegenerateConstraints} as
\begin{equation}
  \label{Eq:AllZeroConstraints}
  \traceout{AB} ~ \vecP{AB} = 1, \qquad \mu_i ~ \vecP{AB} = 0
\end{equation}
where $\traceout{AB} = \traceout{A} \otimes \traceout{B}$, and $\{ \mu_i \}$ is a set of linearly independent forms having value zero for correlations in the considered scenario.
\marginnote{A complete definition of equivalency is given in Section~\ref{Sec:EquivalentBellExpressions}.}
We say that inequalities are equivalent when they define the same hyperplace in the affine space constrained by Eq.~\eqref{Eq:AllZeroConstraints}.
\begin{definition}
  \label{Def:AffineEquivalent}
  With $\traceout{AB}$ and $\{\mu_i\}$ as defined above, the $\Kset{}$-inequalities $(\expr{AB}, u)$ and $(\expr{AB}', u')$ are {\em affine equivalent} if
\begin{equation}
  \expr{AB}' = s ~ (\expr{AB} + t ~ \traceout{AB} + \sum_i w_i ~ \mu_i), \qquad u' = s(u + t),
\end{equation}
for $s>0$ and $t, w_i \in \RR$.
\end{definition}
This definition is sound because the~$\mu_i$ will lead to a zero-contribution and~$\traceout{AB}$ to a constant~$t$.
\subsubsection{Checking equivalence and canonical representatives}
While the above condition can be checked by solving a linear system, we prefer to find a method that takes any $(\expr{AB}, u)$ to a canonical representative.
The idea of a canonical representative is to have an expression that follows some recipe and that is as minimal as possible; {\it e.g.}, we will project away all the components of~$\exp{AB}$ which are irrelevant.
Note that Proposition~\ref{Prop:NonDegenerateConstraints} and its corollary splits the space $\Pspace{A} \otimes \Pspace{B}$ into a normalization subspace $\Zof{AB} = \Zof{A} \otimes \Zof{B}$, allowed subspaces such as $\Cof{A} \otimes \Zof{B}$ and forbidden subspaces such as $\Zof{A} \otimes \Sof{B}$, with the exact grouping depending on the allowed signaling directions.
We group the allowed subspaces (with the exception of $\Zof{A} \otimes \Zof{B}$) into $\Gamma$, while the forbidden subspaces are grouped into $\Omega$; the corresponding dual elements are also grouped into $\Gamma^*$ and $\Omega^*$ such that
\begin{equation}
  \Pspace{A} \otimes \Pspace{B} = \Zof{AB} \oplus \Gamma \oplus \Omega, \qquad   (\Pspace{A} \otimes \Pspace{B})^* = \Zof{AB}^* \oplus \Gamma^* \oplus \Omega^*\;.
\end{equation}
Remark that the linear forms in $\Omega^*$ evaluate to zero on allowed behaviors.
Thus, $\{\mu_i\}$ is a basis of $\Omega^*$.

Recalling the subspace projectors given by Corollary~\ref{Cor:Projectors}, we write $\Pi_\Zsub$ the projector on the normalization subspaces:
\begin{equation}
  \Pi_\Zsub = \Pi_{\Zof{A}} \otimes \Pi_{\Zof{B}}\;.
\end{equation}
We also write $\Pi_\Gamma$ (resp. $\Pi_\Omega$) the projector on the $\Gamma$ subspace (resp. $\Omega$):
\begin{equation}
  \label{Eq:Projectorss}
  \Pi_\Gamma = \Pi_{\Zof{A} \otimes \Cof{B}} + \Pi_{\Cof{A} \otimes \Zof{B}} + \Pi_{\Cof{A} \otimes \Cof{B}} + \ldots, \quad \Pi_\Omega = \Pi_{\Zof{A} \otimes \Sof{B}} + \Pi_{\Sof{A} \otimes \Zof{B}} + \ldots\;.
\end{equation}

The projectors defined so far act on behaviors.
\marginnote{Recall that the adjoint $\Pi_\Zsub^\dagger$ is a fancy name for the row-vector-matrix multiplication action of $\Pi_\Zsub$.}
However, if $\Pi_\Zsub$ has image in the $\Zof{AB}$ subspace, its adjoint operator $\Pi_\Zsub^\dagger: (\Pspace{A} \otimes \Pspace{B})^* \to (\Pspace{A} \otimes \Pspace{B})^*$ has image into $\Zof{AB}^*$.

Then, the canonical form of $\expr{AB}$ is obtained by either:
\begin{itemize}
\item Projecting $\expr{AB}$ to $\expr{AB}' = \expr{AB} ~ \Pi_\Gamma$, and shifting $u' = u - \expr{AB} ~ \Puni{A}\otimes\Puni{B}$.
  This projection is relevant when the same Bell expression is equipped with a variety of bounds corresponding to various convex sets of interest (local, quantum, \ldots).
  The projected Bell expression $\expr{AB}'$ is independent of the bound $u$ considered.
  This is the approach used in our classification of Bell inequalities~\cite{Rosset2014a}.
\item Projecting and shifting $\expr{AB}$ to $\expr{AB}' = \expr{AB} ~ (\Pi_\Zsub + \Pi_\Gamma) - u ~ (\traceout{A}\otimes \traceout{B})$ while $u'=0$.
  The resulting Bell inequality is fully characterized by $\expr{AB}'$, which is an element of the dual of the convex set of interest $\Kset{AB}$.
  The projection depends on the original bound $u$.
\end{itemize}
After projection and shift, we have two options to fix the scale factor in order to define a unique representative.
\begin{itemize}
\item If $\expr{AB}'$ has rational coefficients, fix the scale by writing the expression with relatively prime integers.
\item Otherwise, use an arbitrary norm such that the $1$-norm or $\infty$-norm to set $|| \expr{AB} || = 1$.
\end{itemize}
Note that in both cases, this projection, shift and rescale procedure commutes with relabelings of inputs and outputs (see Section~\ref{Sec:InvariantTechnical:RelabelingInvariant}).

\subsection{Optimizing the variance of inequalities}
\label{Sec:InvariantApplications:Variance}
Consider the use of a Bell expression $\expr{AB}$ on experimental data to evaluate the violation of a Bell inequality.
The experimental data corresponds to a random variable $\vecP{AB}^\text{run}$ with covariance matrix $\Sigma_{(abxy),(a'b'x'y')}$; for example, $\vecP{AB}^\text{run}$ can be obtained by computing the relative frequencies
\begin{equation}
  \operatorname{P}^\text{run}(ab|xy) = \frac{N_{abxy}}{\sum_{ab} N_{abxy}}
\end{equation}
where $N_{abxy}$ are the event counts.
In a previous work~\cite{Renou2016}, we showed that the CH and CHSH Bell inequalities, while equivalent under Definition~\ref{Def:AffineEquivalent}, correspond to different statistical estimators of the Bell expression value.
In particular, while the mean of the estimated value is unchanged when switching between CH and CHSH, the variance differs.

In~\cite[Appendix D]{Renou2016}, we showed that the form of a Bell expression can be tuned to minimize the variance of the random variable $\expr{AB} ~ \vecP{AB}^\text{run}$.
The terms that are tuned correspond to the subspace $\Omega$ spanned by the $\{\mu_i\}$ of Eq.~\eqref{Eq:AllZeroConstraints}.
Our previous work presented a subspace decomposition valid only for the CHSH scenario (two parties, binary inputs and outputs).
By following the same construction, we obtain that the Bell expression $\expr{AB}^\star$ with optimal variance has the form:
\marginnote{In this computation, when the inverse is not defined, it should be replaced by the Moore-Penrose pseudo-inverse.}
\begin{equation}
  \expr{AB}^\star = \expr{AB} (\overline{\Pi}_\Omega - \Pi_\Omega (\Pi_\Omega \Sigma \Pi_\Omega + \overline{\Pi}_\Omega)^{-1} \Pi_\Omega \Sigma \overline{\Pi}_\Omega)\;,
\end{equation}
where $\Pi_\Omega$ has been defined in Eq.~\eqref{Eq:Projectorss} and
\begin{equation}
  \overline{\Pi}_\Omega = \Mid - \Pi_\Omega = \Pi_\Zsub + \Pi_\Gamma\;.
\end{equation}

\subsection{Decomposing steering witnesses and assemblages}
 \label{Sec:InvariantApplications:Steering}
 Consider a two-party steering scenario where Alice is device-independent and Bob device-dependent; for an introduction to the concepts, see the reference~\cite{Cavalcanti2017b}.
The main object of study is an assemblage $\boldsymbol{\sigma} = \{\sigma_{a|x}\}_{a,x}$.
We write $\mathsf{H}_\mathrm{B}$ the space of Hermitian operators for device B, whose underlying Hilbert space has dimension $d$.
We easily identify $\boldsymbol{\sigma} \in \Pspace{A} \otimes \mathsf{H}_\mathrm{B}$.
Steering witnesses $\boldsymbol{F} = \{ F_{a|x} \}$ are elements of the dual space $(\Pspace{A} \otimes \mathsf{H}_\mathrm{B})^*$.
The space $\mathsf{H}_\mathrm{B}$ is an inner product space and thus self-dual.
While we already know the decomposition of $\Pspace{A}$, we need to decompose $\mathsf{H}_\mathrm{B}$.
\begin{lemma}
  Under completely-positive-trace-preserving (CPTP) maps, which are the relevant local transformations for density matrices, the space $\mathsf{H}_\mathrm{B}$ decomposes as
\begin{equation}
  0 \subset \Csub_{\mathrm{B}} \subset \Csub_{\mathrm{B}} \oplus \Zsub_{\mathrm{B}} = \mathsf{H}_\mathrm{B}
\end{equation}
where
\begin{equation}
  \Zsub_{\mathrm{B}} = \{ \alpha \Mid_d \text{ s.t. } \alpha \in \RR \}, \qquad \Csub_{\mathrm{B}} = \{ \rho \in \mathsf{H}_\mathrm{B} \text{ s.t. } \operatorname{tr}\rho = 0 \}\;.
\end{equation}
\end{lemma}
\begin{proof}
  Decompose $\mathsf{H}_\mathrm{B}$ as a representation of unitary maps, which form a subset of CPTP maps, to obtain a subspace spanned by $\Mid$ and a subspace corresponding to traceless matrices.
  Invariant subspaces of CPTP maps are necessarily coarser; we easily verify that $\Zsub_{\mathrm{B}}$ is not invariant under nonunital CPTP maps, while $\Csub_{\mathrm{B}}$ is an invariant subspace.
\end{proof}
Note that $\Zsub_{\mathrm{B}}$ and $\Csub_{\mathrm{B}}$ are orthogonal subspaces.
Consider now the decomposition of $\Pspace{A} \otimes \mathsf{H}_\mathrm{B} = (\Cof{A} \oplus \Zof{A} \oplus \Sof{A}) \otimes (\Csub_\mathrm{B} \oplus \Zsub_\mathrm{B})$, we obtain the following result.
\begin{proposition}
  A quantum assemblage has the form
  \begin{equation}
    \boldsymbol{\sigma} = \frac{1}{d} \Puni{A} \otimes \Mid_d + \boldsymbol{\Delta}, \qquad \boldsymbol{\Delta} \in (\Zof{A} \otimes \Csub_\mathrm{B}) \oplus (\Cof{A} \otimes \Zsub_\mathrm{B}) \oplus (\Cof{A} \otimes \Csub_\mathrm{B})\;.
  \end{equation}
\end{proposition}
\begin{proof}
  By definition, quantum assemblages satisfy
  \begin{equation}
    \sum_a \sigma_{a|x} - \sigma_{a|x'} = 0, \qquad \frac{1}{X} \sum_{a,x} \sigma_{a|x} = \rho_\mathrm{B}, \qquad \frac{1}{X} \sum_{a,x} \operatorname{tr}~\sigma_{a|x} = 1\;.
  \end{equation}
  The leftmost constraint translates to $(\Sanh{A}^i \otimes W) ~ \boldsymbol{\sigma} = 0$ for all $W \in \mathsf{H}_\mathrm{B}^*$ and implies the middle constraint.
These constraints forbid elements in the $\Sof{A} \otimes (\Csub_\mathrm{B} \oplus \Zsub_\mathrm{B})$ space.
The rightmost constraint translates to $(\Sanh{A}^i \otimes \Mid) ~ \boldsymbol{\sigma} = 1$ where $\Mid$ is interpreted as the linear form $\Mid(\rho) = \operatorname{tr} \rho$, and fixes the element in $\Zof{A} \otimes \Zsub_\mathrm{B}$ to $(\Puni{A} \otimes \Mid_d)/d$.
\end{proof}

Then, as for Bell-like inequalities, the bound of steering witnesses can be shifted using the component in $(\Zof{A} \otimes \Zsub_{\mathrm{B}})^*$, and steering witnesses are equivalent under the addition of arbitrary terms in $(\Sof{A} \otimes \mathsf{H}_\mathrm{B})^*$.

\clearpage

\section{Proofs and technical details}
\label{Sec:InvariantTechnical}

This Section contains technical details and proofs.
While it should be skipped at first reading, it reveals the underlying algebraic structure of local transformations.

\subsection{Preliminary definitions}

Compared to our early work~\cite{Renou2016} that used representation theory of finite groups, we have two complications in the present study.
First, the space $\Pspace{A}$ is not an inner product space.
Second, the set of local transformations is not a group: It contains irreversible transformations.
We start by establishing the formalism in which we provide our proof, following mostly the lecture notes~\cite{Etingof2009}.

\subsubsection{Algebraic structure of local transformations}
Our decomposition is based on representation theory, where the matrices $\Lmap{A}$ are representations of objects we now identify.

First, recall that any local transformation $\Lmap{A}$ can be written as a convex combination of deterministic local maps.
Deterministic local maps are described by the pair $s = (\xi, \overline{\alpha})$ according to Definition~\ref{Def:LocalDeterministicMap}.
\begin{proposition}
  \label{Prop:DetClosed}
  The set of deterministic local maps $\Pspace{A} \to \Pspace{A}$, represented by pairs $s = (\xi, \overline{\alpha})$, is closed under composition (given in  Section~\ref{Sec:LiftingDefinitions:Composition}), has an identity element $e$ (given by having all $\xi$, $\alpha_{x'}$ identity maps themselves). It is thus a monoid which we write $\Det{A}$.
\end{proposition}

\marginnote{An alternative route for the present study is to consider the irreducible representations of the monoid $\Det{A}$ according to~\cite{Steinberg2016}.
  However, the study of those representations requires more knowledge about the structure of $\Det{A}$ than the pedestrian approach we develop here. 
}
As local transformations are convex (and thus linear) mixtures of deterministic transformations, we consider the set
\begin{equation}
  \mathsf{L} = \{ L : L = \sum_{s\in \Det{A}} \lambda_s \cdot s, \quad \lambda_s \in \mathbb{R} \}
\end{equation}
of formal sums which form a real vector space.
We add the composition rule, for $L = \sum_{s\in \Det{A}} \lambda_s \cdot s $ and $M = \sum_{t\in\Det{A}} \mu_t \cdot t$:
\begin{equation}
  L \cdot M = \sum_{s,t\in\Det{A}} (\lambda_s \mu_t) ( s \circ t) \;.
\end{equation}

Using Definition~\ref{Def:LocalDeterministicMap}, each $s$ corresponds to a matrix $\Lmap{A}(s)$.
The formal sum $L$ corresponds to a real matrix $\Lmap{A}(L)$
\begin{equation}
  \Lmap{A}(L) = \sum_{s\in\Det{A}} \lambda_s ~ \Lmap{A}(s) \; .
\end{equation}
When $\sum_s \lambda_s = 1$ and $\lambda_s \ge 0$, the resulting $\Lmap{A}(L)$ is a local transformation by Proposition~\ref{Prop:LocalTransformationsConvex}.
However, the same $\Lmap{A}$ can correspond to different convex decompositions $L$.
\marginnote{In that case, by a slight abuse of terminology, we say that $\Vof{A}$ is a representation of $\mathsf{L}$.}

With this construction, the set $\mathsf{L}$ is an associative algebra.
By linearity and as $\Lmap{A}(L ~ M) = \Lmap{A}(L) ~ \Lmap{A}(M)$, this algebra has a representation $L \mapsto \Lmap{A}(L)$ on the vector space $\Pspace{A}$.

\subsubsection{Invariant subspaces, subrepresentations and derived representations}
We now describe subrepresentations of $\Pspace{A}$ and representations that can be derived from subrepresentations.
\begin{definition}
  A subspace $\Vsub \subset \Pspace{A}$ is an {\em invariant subspace} under the maps $\left \{ \Lmap{A} : \Pspace{A} \to \Pspace{A} \right \}$ if the following holds: For all $\Lmap{A}$ and all $\vec{v} \in \Vsub$, the image $(\Lmap{A} ~ \vec{v})$ is in $\Vsub$.
\end{definition}
The invariant subspace $\Vsub$, through the restriction of $\Lmap{A}: \Pspace{A} \to \Pspace{A}$ to $\Vsub \to \Vsub$ is a subrepresentation of the associative algebra $\mathsf{L}$.
If $\Vsub$, in turn, has no nontrivial invariant subspace, then $\Vsub$ is an {\em irreducible} representation of $\mathsf{L}$.

Given a subrepresentation of $\Pspace{A}$, we can generate other representations of $\mathsf{L}$.
First, we look at the invariant subspaces of $\Pspace{A}^*$ under the action of the adjoint $\Lmap{A}^\dagger$.
\marginnote{Example: Consider a device with binary inputs and outputs $\pst{A} = (2,2)$.
  Consider the subspace $\Vsub$ spanned by all the deterministic behaviors $\vec{v}_1 = (1,0,1,0)^\top$, $\vec{v}_2 = (1,0,0,1)^\top$, $\vec{v}_3 = (0,1,1,0)^\top$, $\vec{v}_4 = (0,1,0,1)^\top$.
  This subspace has dimension three and is invariant under local transformations.
  Its annihilator is spanned by $\nu = (1,1,-1,-1)$ and is also invariant under local transformations; it corresponds to the form that checks whether a distribution has the same normalization factor accross inputs.
  Recall that the order of coefficient enumeration has been discussed in Section~\ref{Sec:Definition:Behaviors} and the annihilator has been informally discussed in Eq.~\eqref{Eq:InvariantSubspaces:Motivation:Ann}.
}
For a subspace $\Vsub \subset \Pspace{A}$, we define the {\em annihilator space} $\Vsub^0 \subset \Pspace{A}^*$
\begin{equation}
  \Vsub^0 = \left \{ \expr{} \in \Pspace{A}^* \text{ s.t. } \expr{} ~ \vec{v} = 0 \text{ for all } \vec{v} \in \Vsub \right \}
\end{equation}
which has dimension $\dim \Vsub^0 = \dim \Pspace{A} - \dim \Vsub$.
\begin{proposition}
  The space $\Vsub^0$ is an invariant subspace of $\Pspace{A}^*$.
\end{proposition}
\begin{proof}
Let $\expr{} \in \Vsub^0$.
Then $(\expr{} ~ \Lmap{A})$ is in $\Vsub^0$ as well: For all $\vec{v} \in \Vsub$, we have $(\expr{} ~ \Lmap{A}) ~ \vec{v} = 0$ as $(\Lmap{A} ~ \vec{v}) \in \Vsub$.
\end{proof}

The quotient space $\Pspace{A} / \Vsub$ is defined as the set of equivalence classes of the relation:

\begin{equation}
  \vec{x} \equiv_{\Vsub} \vec{y} \qquad \text{if} \qquad \vec{x} - \vec{y} \in \Vsub\;.
\end{equation}

It is a vector space but not a subspace of $\Pspace{A}$.
To emphasize that its elements correspond to elements of $\Pspace{A}$ up to an element of $\Vsub$, we write these elements as $\vec{w} + \Vsub$ for $\vec{w} \in \Pspace{A}$.
\begin{proposition}
  The quotient space $\Pspace{A} / \Vsub$ is a representation of $\mathsf{L}$.
\end{proposition}
\begin{proof}
  For all $\vec{v} \in \Vsub$, we have $\Lmap{A}(L) \big( \vec{w} + \vec{v} \big) = \Lmap{A}(L) ~ \vec{w} + \underbrace{\Lmap{A}(L) ~ \vec{v}}_{\equiv 0}$.
\end{proof}

It is easier to study $\Pspace{A} / \Vsub$ through an explicit basis of its annihilator $\Vsub^0$.
\begin{lemma}
  \label{Lemma:QuotientAnnihilator}
Let $\Vsub$ be a representation of $\mathsf{A}$ and $\Vsub^0$ its annihilator.
Let $\{ \nu_i \}_{i=1}^d$ be a basis of $\Vsub^0$ (let us remind that those basis elements are linear forms).
We consider the map $f: \Pspace{A} \to \RR^d$ that evaluates those linear forms:
\begin{equation}
  f: \vec{v} \mapsto \vec{q} = (q_1, \ldots, q_d)\;,
\end{equation}
where the $q_i$ are computed according to $q_i = \nu_i ~ \vec{v}$.
Then the image of $f$ is isomorphic to $\Pspace{A} / \Vsub$, and affords a representation of $\mathsf{L}$.
\end{lemma}
\begin{proof}
  See~\cite[Theorem 3.16]{Roman2005}.
\end{proof}

\subsubsection{Filtrations}

We are now looking at decompositions of the space $\Pspace{A}$ and its dual.
\begin{definition}
  A {\em filtration} is a chain of subspaces
  \begin{equation}
    0 = \mathsf{V}_0 \subset \mathsf{V}_1 \subset \ldots \subset \mathsf{V}_n = \Pspace{A}
  \end{equation}
  such that each $\mathsf{V}_i$ is invariant and affords a subrepresentation of the algebra $\mathsf{L}$.
\end{definition}

\marginnote{We denote the zero vector space by $0$.}
A filtration provides a decomposition of the dual space $\Pspace{A}^*$.
\begin{lemma}
  \label{Lemma:AnnihilatorDualDecomposition}
  Let $0 = \mathsf{V}_0 \subset \mathsf{V}_1 \subset \ldots \subset \mathsf{V}_n = \Pspace{A}$ be a chain of invariant subspaces of $\Pspace{A}$ under $\Lmap{A}$.
  Then
  \begin{equation}
    0 = \Vsub_n^0 \subset \Vsub_{n-1}^0 \subset \ldots \subset \Vsub_1^0 \subset \Vsub_0^0 = \Pspace{A}^*
  \end{equation}
  is a chain of invariant subspaces of $\Pspace{A}^*$.
\end{lemma}
\begin{proof}
  The invariance of $\Vsub_i^0$ has already been proved above.
  For $\Vsub_{i-1} \subset \Vsub_i$ implies $\Vsub_i^0 \subset \Vsub_{i-1}^0$, see~\cite[Theorem 3.14]{Roman2005}.
\end{proof}

\subsection{Subspace decomposition}

We now move to the proof of the uniqueness of the decomposition presented in Proposition~\ref{Prop:Decomposition}.
According to~\cite[Lemma 2.8]{Etingof2009}, every finite dimensional representation admits a finite filtration such that the successive quotients $\mathsf{V}_i / \mathsf{V}_{i-1}$ are irreducible, and according to the Jordan-H{\"o}lder theorem~\cite[Section 2.7]{Etingof2009}, this filtration is unique up to the permutation of subspaces.

While it is not difficult to verify that the stated invariant subspaces in Proposition~\ref{Prop:Decomposition} are indeed invariant (at least for concrete cases), proving the irreducibility of the successive quotients is more involved.

\begin{proposition}
  \label{Prop:Filtration}[Technical version of Proposition~\ref{Prop:Decomposition}]
  The representation $\Lmap{A}(\cdot)$ of $\mathsf{L}$ admits the filtration
  \begin{equation}
    0 \subset \Cof{A} \subset \Cof{A} \oplus \Zof{A} \subset \Cof{A} \oplus \Zof{A} \oplus \Sof{A} = \Pspace{A}\;,
  \end{equation}
  such that the successive quotients are irreducible.
\end{proposition}

We now prove Proposition~\ref{Prop:Filtration}.
\subsubsection{Irreducibility of $\Cof{A}$}
We first show that $\Cof{A}$ is a subrepresentation of $\mathsf{L}$.
By linearity, it is sufficient to consider the action of elements of $\Det{A}$ on $\Cof{A}$.
Let $s = (\xi, \overline{\alpha}) \in \Det{A}$ be an abstract deterministic map and $\vecP{A} \in \Cof{A}$ while $\vecP{A}' = \Lmap{A}(s) ~ \vecP{A}$.
By definition, $\Cof{A}$ is the maximal subspace of $\Pspace{A}$ such that $\sigma_x ~ \vecP{A} = 0$ for all $x$.
We verify easily that $\sigma_{x'} ~ \vecP{A}' = \sigma_{\xi(x')} ~ \vecP{A} = 0$, and thus $\vecP{A}' \in \Cof{A}$.
Thus $\Cof{A}$ is a subrepresentation of $\Pspace{A}$.

We assume that at least one $A_x > 1$ such that the subspace $\Cof{A}$ is nonzero.
We show that this space $\Cof{A}$ is irreducible using the following lemma, noting that $\Cof{A}$ is spanned by the vectors $\data{A}^{i|j}$ defined in Eq.~\eqref{Eq:CorrelationVectors}.

\begin{lemma}
  Given a nonzero $\vecP{A} \in \Cof{A}$, we can construct a basis vector $\data{A}^{i|j}$ for arbitrary $i$, $j$ using an appropriate deterministic map in $s = (\xi, \overline{\alpha}) \in \Det{A}$, such that
\begin{equation}
  \Lmap{A}(s) ~ \vecP{A} = w ~ \data{A}^{i|j}
\end{equation}
for some $w \neq 0$.

\end{lemma}
\begin{figure}
  \begin{center}
    \includegraphics{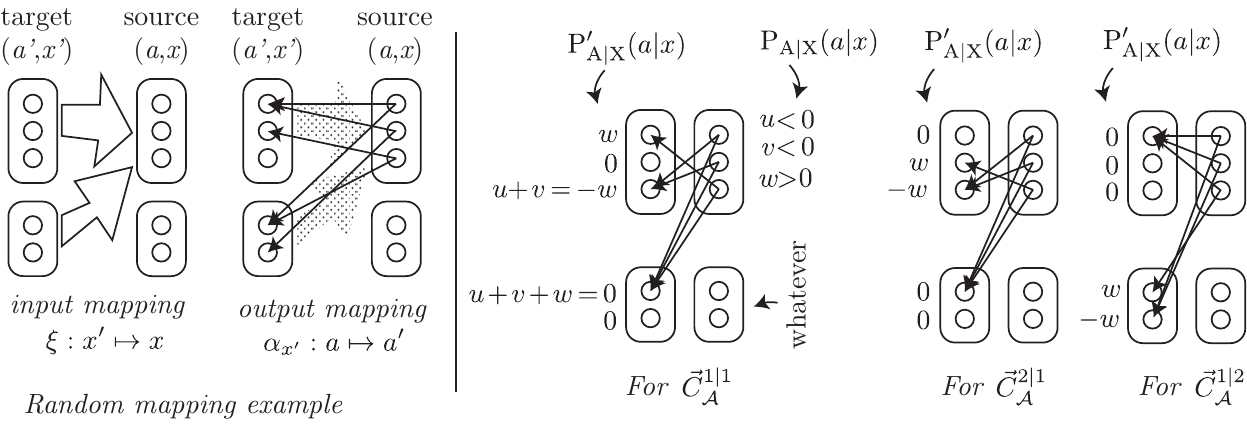}
  \end{center}
  \caption{
    \label{Fig:InvarianceC}
    We provide an illustration of the irreducibility of the subspace $\Cof{A}$ for the party structure $\pst{A} = (3, 2)$, and thus for the space $\Pspace{A} = \RR^5$.
    This figure illustrates the notation used in the later Section~\ref{Sec:LiftingDefinitions}.
    Left: The abstract deterministic local maps are composed of a mapping of inputs $\xi$ (white arrows) and, given this input mapping, a mapping of outputs $\alpha_{x'}$ (black arrows).
    Right: We consider a source vector $\vec{c} = (u, v, w, \zeta_0, \zeta_1)$ where $\zeta_0$, $\zeta_1$ are arbitrary and some of the $u$, $v$, $w$ are nonzero.
    In our example, we have $u,v<0$ and $w>0$ while $u+v+w = 0$ by definition.
    We then display the deterministic local maps that construct the three basis vectors $\data{A}^{1|1}$, $\data{A}^{2|1}$ and $\data{A}^{1|2}$.
  }
\end{figure}
\begin{proof}
The proof is straightforward.
As $\sum_a \Pgivenop{A}{X}{a}{x} = 0$ but $\vecP{A} \ne 0$, at least one input $x = \nu$ has some nonzero coefficients $\{\Pgivenop{A}{X}{a}{\nu}\}_a$.
We use the input mapping $\xi(x') = \nu$ for the deterministic map, and consider  now the output mappings $\alpha_{x'}: \domain{A_\nu} \to \domain{A_{x'}}$.
For $x' \neq j$, we set $\alpha_{x'}(a) = 1$, which sets the coefficients for $(a',x' \neq j)$ to zero. For $x' = j$, we fix the input $x = \nu$ and consider for which $a$ the source vector $\vecP{A}$ has corresponding positive, negative or zero coefficients.
If the sum of the positive coefficients is $w_+$ and the sum of the negative coeficients is $w_-$, by $\sigma_\nu(\vecP{A}) = w_+ + w_- = 0$ we have $w_+ = - w_-$.
Then, to form the basis vector $\data{A}^{i|j}$, it is sufficient to send the positive-valued outputs to $a' = i$, and the negative-valued outputs to $a' = A_j$, while the zero-valued outputs can be sent anywhere.
The scheme is also illustrated in Figure~\ref{Fig:InvarianceC}.
\end{proof}

\subsubsection{Decomposition of $\Pspace{A} / \Cof{A}$}
\label{Sec:Decomposition1}
We have now the chain $0 \subset \Cof{A} \subset \Pspace{A}$, and study whether the quotient $\Pspace{A} / \Cof{A}$ is irreducible.
As the $\{\sigma_x\}_x$ provide a basis of the annihilator $\Cof{A}^0$, we use Lemma~\ref{Lemma:QuotientAnnihilator} with the observation that $\sigma_{x'} ~ \vecP{A}' = \sigma_{\xi(x')} ~ \vecP{A}$ under local transformations.
The subspace of elements with homogenous normalization, $\sigma_x ~ \vecP{A} = \omega$ for all $x$, is invariant under $\Det{A}$.
It is thus sufficient to complement $\Cof{A}$ with a single vector $\vecP{A}$ obeying $\sigma_x ~ \vecP{A} = \omega$ for all $x$ to obtain a new invariant subspace.

However, there is some freedom in the choice of this new vector.
We make a choice here to proceed with the proof, the motivations will become clear in Section~\ref{Sec:InvariantTechnical:FixingDOF}: We require $\Zof{A}$ to be invariant permutation of outputs, and then only choice is a vector proportional to the uniformly random distribution $\Puni{A}$.
We fix the scaling so that $\Puni{A}$ is a properly normalized probability distribution.

Note that the space $\Zof{A}$ spanned by $\Puni{A}$ is {\em not} invariant under deterministic local maps.
However, the space $\Cof{A} \oplus \Zof{A}$ is invariant, as is the quotient space $(\Cof{A} \oplus \Zof{A}) / \Cof{A}$ given by $\omega \Puni{A} + \Cof{A}$ for $\omega \in \RR$.
This quotient space has dimension one and thus corresponds to an irreducible (trivial) representation.
As such, it cannot be split further.

\subsubsection{Decomposition of $\Pspace{A} / (\Cof{A} \oplus \Zof{A})$}
We verify that the quotient space $\mathcal{Q} = \Pspace{A} / (\Cof{A} \oplus \Zof{A})$ is irreducible.
A basis of the annihilator $(\Cof{A} \oplus \Zof{A})^0$ is given by the maps $\Sanh{A}^i \in \Pspace{A}^*$, for $i = 1,\dots,X-1$.
Thus, $\mathcal{Q}$ is isomorphic to $\RR^{X-1}$ through the map
\begin{equation}
  \label{Eq:IsomorphismSQuotient}
  f: \mathcal{Q} \to \RR^{X-1}, \quad \vec{q} \mapsto \vec{v} = (\Sanh{A}^1 ~ \vec{q}, \dots, \Sanh{A}^{X-1} ~ \vec{q}) \;,
\end{equation}
as discussed in Lemma~\ref{Lemma:QuotientAnnihilator}.
When applying a local transformation to an element of $\mathcal{Q}$, only the input mapping $\xi$ modifies $\vec{v}$, the output mapping $\overline{\alpha}$ is irrelevant as $\Sanh{A}^i$ sums over all output values.
Thus, we study the action of local transformation through their input mappings only.
We now show that $\mathcal{Q}$ is irreducible, and start with any nonzero element $\vec{v}$.
Using suitable permutations, we can transform $\vec{v}$ such that its first two coefficients obey $v_1 \ne v_2$.
We then consider local transformation where the input mappings preserve the last input: $\xi(X) = X$.
Then, the action on $\vec{v}$ is such that:
\begin{equation}
  \vec{v}' = \restr{\Lmap{A}(s)}{\mathcal{Q}} ~ \vec{v} \quad \implies \quad {v'}_{i'} = v_{\xi(i')}\; \text{ for } i'=1,\dots,X-1\;.
\end{equation}
Now, for $i = 1, \dots, X-1$, we define the input mapping family $\{ \xi_i \}_i$
\begin{equation}
  \xi_i(x') = \begin{cases} 1 & \mbox{if } x' = i \\ 2 & \mbox{if } x' \neq i \mbox{ and } x' < X \\ X & \mbox{if } x' = X \end{cases}
\end{equation}
which generates the vectors $(v_1, v_2, \dots, v_2)$ and $(v_2, v_1, \dots, v_2)$ until $(v_2, v_2, \dots, v_1)$, which together provide a basis of $\RR^{X-1}$.
This shows that $\mathcal{Q}$ is irreducible, which completes the decomposition.
Note that the kernel of the representation $\restr{\Lmap{A}(s)}{\mathcal{Q}}$ is the set of all deterministic maps with an input mapping $\xi$ corresponding to the identity.

\subsection{Unicity of the decomposition}
By the Jordan-H{\"o}lder theorem~\cite[Theorem~2.18]{Etingof2009}, the decomposition is unique up to permutation of quotients.
In the decomposition above, we picked $\Cof{A}$ as an irreducible representation of $\mathsf{L}$ at the start of the chain.
Let us look now at the other two candidates to place at the start of the chain through their kernels.

The quotient $(\Cof{A} \oplus \Zof{A})/\Cof{A}$ corresponds to a trivial representation, its kernel is the set of all deterministic transformations $\Det{A}$.
The kernel of the quotient $\Pspace{A} / (\Cof{A} \oplus \Zof{A})$ is the set of deterministic maps with identity input mapping, which we write $\mathsf{Out}_\mathcal{A} \subset \Det{A}$.
However, it is impossible to find an invariant subspace of $\Pspace{A}$ whose kernel contains $\mathsf{Out}_\mathcal{A}$ (except in the pathological case where all $A_x = 1$, but then the dimension of $\Cof{A}$ would be zero anyway).
Thus $\Cof{A}$ has to be the first subspace in the chain.

The remaining freedom is whether we can find another irreducible representation in the quotient $\Pspace{A} / \Cof{A}$.
As observed in Section~\ref{Sec:Decomposition1}, the annihilator of $\Cof{A}$ has basis $\{\sigma_x\}_{x=1,\dots,X}$.
Thus, the quotient space is isomorphic to $\RR^X$ through
\begin{equation}
  h: \Pspace{A} / \Cof{A} \to \RR^X, \quad \vecP{A} + \Cof{A} \mapsto \vec{u} = (\sigma_1(\vecP{A}), \dots, \sigma_X(\vecP{A})) \;,
\end{equation}
and deterministic maps transform $\vec{u} \to \vec{u}'$ by mapping its coefficients $u'_{x'} = u_{\xi(x')}$.
Now, note that any invariant subspace of that quotient space has to include the vector $\vec{u}_0 = (1,\dots,1)$, as it is proportional to the result of the action of a deterministic map with $\xi(x') = \text{cte}$.
The subspace spanned by $\vec{u}_0$ is itself invariant, and corresponds to the trivial representation of $\Det{A}$, which has to be the next quotient space in the composition series.
This completes the proof of the unicity of the decomposition.

\subsection{Fixing the degrees of freedom afforded by the Jordan-H{\"o}lder theorem}
\label{Sec:InvariantTechnical:FixingDOF}
We used the Jordan-H{\"o}lder theorem to prove the uniqueness of our decomposition in a chain of subspaces
\begin{equation}
  0 \subset \Cof{A} \subset \Cof{A} \oplus \Zof{A} \subset \Cof{A} \oplus \Zof{A} \oplus \Sof{A} = \Pspace{A}
\end{equation}
but does not prescribe the form of the $\Zof{A}$ and $\Sof{A}$ subspaces; nor it does prescribe a particular convention for the basis vectors used to construct those subspaces.
Thus, how do we motivate the convention proposed in Section~\ref{Sec:InvariantSubspaces:Correlators}?

\subsubsection{Defining the subspaces $\Zof{A}$ and $\Sof{A}$}
\label{Sec:InvariantTechnical:RelabelingInvariant}
We first consider the problem of singling out the subspaces $\Zof{A}$ and $\Sof{A}$.
For that, we use two principles:
\begin{itemize}
\item The subspaces $\Zof{A}$ and $\Sof{A}$ should be invariant under any permutation of outputs, as the labeling of outputs has no physical relevance (see~\cite{Renou2016}).
\item The trace out form $\traceout{A}$ should correspond to the computation of a marginal probability distribution using a uniformly random distribution of inputs $\operatorname{P}_{\mathrm{X}}(x) = 1/X$, where $X$ is the number of inputs, as the labeling of inputs has no physical relevance.
\end{itemize}
Invariance under permutation of outputs fixes the subspace $\Zof{A}$ as already discussed in Section~\ref{Sec:Decomposition1}.
The subspace $\Sof{A}$, of dimension $X-1$, is mostly determined by invariance under output permutation, which leaves $X$ degrees of freedom.
To remove the last degree of freedom, we use the duality relation $\traceout{A} ~ \vecP{A} = 0$ for $\vecP{A} \in \Sof{A}$, with the form $\traceout{A}$ fixed by the second principle.

\subsubsection{Choice of basis elements}
The remaining freedom to fix in Section~\ref{Sec:InvariantSubspaces:Correlators} is the choice of the particular basis elements.
We use the following guiding principles:
\begin{enumerate}
\item We reuse existing conventions as much as possible.
  In the case of binary outputs, our notation should be compatible with binary correlators.
\item The basis conversion matrices have straightforward structure and are written using rational coefficients with small numerators/denominators.
\item Pure signaling correlations (for example with $b=x$) have coefficients in the corresponding signaling subspace equal to the identity matrix (interpreting the matrix rows as the source space $\Ssub$ and columns as the target space $\Csub$).
\end{enumerate}

The correlation vectors $\data{A}^{i|x}$ are such that the dual elements $\datad{A}^{i|x}$ correspond to the generalized correlators presented in~\cite[Appendix]{Bancal2010}, which satisfies 1. and 2.
The vector $\Puni{A}$ is fixed by normalization.
The signaling vectors $\sig{A}^{i}$ are then chosen to satisfy 3.

\subsection{Proof of Proposition~\ref{Prop:Normalization}}
\label{Sec:ProofPropNormalization}
First, remark that normalization prescribes:
\begin{equation}
  \sum_{a_1a_2\ldots b_1b_2\ldots} \Pgivenop{A_1A_2\ldots B_1B_2\ldots}{X_1X_2\ldots Y_1Y_2\ldots}{a_1a_2\ldots b_1b_2 \ldots}{x_1x_2\ldots y_1y_2 \ldots} = 1\;,
\end{equation}
for all $x_1,x_2,\ldots,y_1,y_2,\ldots$, which is equivalent to
\begin{equation}
  (\SsumI{A}{1}^{x_1} \otimes \SsumI{A}{2}^{x_2} \otimes \ldots \otimes \SsumI{B}{1}^{y_1} \otimes \SsumI{B}{2}^{y_2} \otimes \ldots) ~ \vecP{} = 1\;,
\end{equation}
for the same indices.
We remark that $\{\SsumI{A}{1}^{x_1}\}_{x_1}$ spans the same subspace as $\{\traceoutI{A}{1}\} \cup \{\SanhI{A}{1}^{k_1}\}_{k_1}$.
Thus we rewrite the above constraint either as:
\begin{equation}
  (\traceoutI{A}{1} \otimes \traceoutI{A}{2} \otimes \ldots \otimes \traceoutI{A}{n_\Zsub}) ~ \vecP{} = 1
\end{equation}
when $n_\Ssub = 0$ or
\begin{equation}
  (\traceoutI{A}{1} \otimes \traceoutI{A}{2} \otimes \ldots \otimes \SanhI{B}{1}^{l_1} \otimes \SanhI{B}{2}^{l_2} \otimes \ldots) ~ \vecP{} = 0
\end{equation}
for all $l_1, l_2, \ldots$ when $n_\Ssub > 0$, and the r.h.s. value is obtained by substituting the definitions~\eqref{Eq:Traceout1} and~\eqref{Eq:DefSanh}.

\subsection{Proof of Proposition~\ref{Prop:Nonsignaling}}
\label{Sec:ProofPropNonsignaling}
  Due to the existence of a $\Csub_{\mathcal{C}_1}$ subspace, we cannot have the cardinality $(C_1)_{z_1} = 1$ for all $z_1$.
  We assume that the output cardinalities for $z_1 = 1$ is $(C_1)_{z_1=1} > 1$; when this is not true, the proof is adapted by replacing $z_1 = 1$ by one of the inputs $z_1$ that has cardinality $(C_1)_{z_1} > 1$.
  The same assumption is made about $\mathcal{C}_2$ and so on.
  
  We consider the deterministic nonsignaling behavior
  \begin{multline*}
    \Pgivenop{A_1 A_2 \ldots C_1 C_2 \ldots}{X_1 X_2 \ldots Z_1 Z_2 \ldots}{a_1 a_2 \ldots c_1 c_2 \ldots}{x_1 x_2 \ldots z_1 z_2 \ldots} = \\
    \Pgivenop{A_1}{X_1}{a_1}{x_1} \Pgivenop{A_2}{X_2}{a_2}{x_2} \ldots \Pgivenop{C_1}{X_1}{c_1}{x_1} \Pgivenop{C_2}{X_2}{c_2}{x_2} \ldots    
  \end{multline*}
  where each single party distribution is deterministic such that $a_1 = a_2 = \ldots = c_1 = c_2 = \ldots = 1$.
  To prove that $\vec{P}_{\mathcal{A}_1\mathcal{A}_2\ldots\mathcal{C}_1\mathcal{C}_2\ldots}$ has support in the aforementionned subspace, we have
\begin{equation}
  (\traceoutI{A}{1} \otimes \traceoutI{A}{2} \otimes \ldots \otimes \datadI{C}{1}^{1|1} \otimes \datadI{C}{2}^{1|1} \otimes \ldots) ~ \vec{P} = \left[ (C_1)_{z_1=1} - 1 \right ] \left [ (C_2)_{z_2=1} - 1 \right ] \ldots  \ne 0\;,
\end{equation}
as $\traceoutI{A}{1} ~ \vec{P}_{\mathcal{A}_1} = 1$, \ldots, and
\begin{equation}
  \datadI{C}{1}^{1|1} ~ \vec{P}_{\mathcal{C}_1} = (C_1)_{z_1=1} - 1
\end{equation}
by~\eqref{Eq:Defdatad}.
Now, the proposition stated that {\em any} deterministic behavior has support in the considered subspace.
Due to the tracing out, the deterministic value of the outputs $a_1, \ldots, a_{n_\Zsub}$ do not impact the proof.
Nevertheless, we assumed that $c_1 = \ldots = c_{n_\Csub} = 1$.
This does not lose generality.
We use a relabeling of outputs to bring the outputs to $c_1 = \ldots = c_{n_\Csub} = 1$.
As the subspace considered is invariant under local transformations, and the transformation is reversible, the proposition follows.

\subsection{Proof of Proposition~\ref{Prop:Signaling}}
\label{Sec:ProofPropSignaling}
We remind Definition~\ref{Def:Nonsignaling}, and write after summing over $\overline{x}$:
\begin{equation}
  \left[\sum_{\overline{x}\overline{a}\overline{b}} \Pgivenop{\overline{A}\overline{B}\overline{C}}{\overline{X}\overline{Y}\overline{Z}}{\overline{a} \overline{b} \overline{c}}{\overline{x} \overline{y} \overline{z}} \right]
  -
  \left[\sum_{\overline{x}\overline{a}\overline{b}} \Pgivenop{\overline{A}\overline{B}\overline{C}}{\overline{X}\overline{Y}\overline{Z}}{\overline{a} \overline{b} \overline{c}}{\overline{x} \overline{y}' \overline{z}} \right] = 0, \quad \forall \overline{c},\overline{y}, \overline{y}', \overline{z}\;.
\end{equation}
Fixing all $\overline{y}'$ to the last input value, we get
\begin{equation}
  (\traceoutI{A}{1} \otimes \traceoutI{A}{2} \otimes \ldots \otimes \SanhI{B}{1}^{y_1} \otimes \SanhI{B}{2}^{y_2} \otimes \ldots \mathbbm{1}_{\mathcal{C}_1} \otimes \mathbbm{1}_{\mathcal{C}_2} \otimes \ldots ) ~ \vec{P} = 0\;.
\end{equation}
This is a vector equation as we left the subspaces $\mathcal{C}_1$, $\mathcal{C}_2$, \ldots unaffected.
We get the proposition by replacing the identity maps by the $\datad{\ldots}$, which corresponds to the subspaces not already covered by Propositions~\ref{Prop:Normalization} and~\ref{Prop:Nonsignaling}.

\clearpage

\part{Liftings}
\label{Part:Liftings}
We now start the third part of our manuscript, and study the reversibility of local transformations.
In particular, we link invertible transformations to the liftings of Bell expressions presented by Pironio~\cite{Pironio2005}, in which specific local transformations process Bell expressions to create new expressions in scenarios with additional inputs and/or outputs.
When the original expression corresponds to a facet of the local polytope, the new expression also corresponds to a local facet.
This implies that some local facets in scenarios of complex structure test actually correlations with a simpler structure, so these Bell expressions are preferably studied in their simpler form.
The present section expands on~\cite{Pironio2005} in two directions.
We prove that the transformations listed in~\cite{Pironio2005} are exhaustive, and generalize them to signaling scenarios.
We also study liftings of behaviors, for example of nonsignaling boxes.

This part of our manuscript is structured as follows.
First, in Section~\ref{Sec:LiftingDefinitions}, we provide an overview and the relevant definitions.
Second, in Section~\ref{Sec:LiftingBehaviors}, we discuss reversible transformations of behaviors, which correspond to liftings of boxes.
Finally, Section~\ref{Sec:LiftingExpressions} addresses reversible transformations of Bell expressions, which corresponds to liftings of Bell-like inequalities.

For simplicity, the arguments in this part are presented for a nonsignaling two-party scenario,
as the generalization to the multi-party case is straightforward.
To avoid prime symbols burdening the notation, we use liberally the letters A, B, C, D, E, F.
The context easily identifies which particular subsystem the spaces $\Pspace{A}, \Pspace{B}, \ldots$ are attached to.

\clearpage

\section{Properties of deterministic local transformations}
\label{Sec:LiftingDefinitions}
\marginnote{
  \begin{center}
    \includegraphics{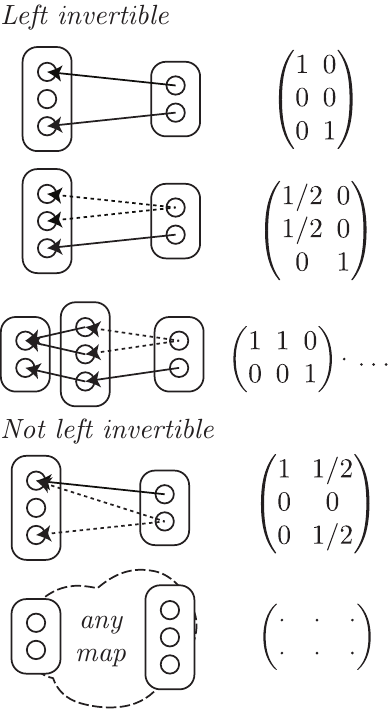}
  \end{center}
  \captionof{figure}{\label{Fig:LeftInvertible}
    Examples of transformations and left invertibility.
    A solid arrow corresponds to a transition probability of $1$, while dotted arrows correspond to $1/2$.
  }
  \begin{center}
    \includegraphics{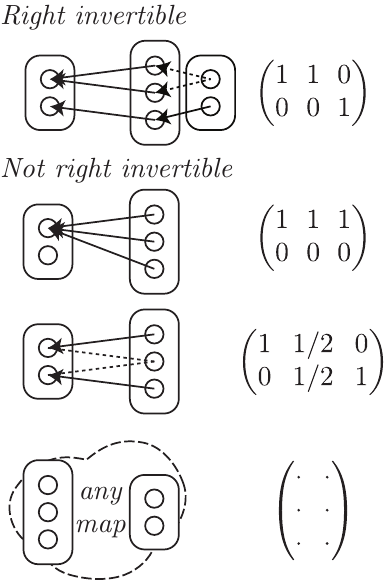}
  \end{center}
  \captionof{figure}{\label{Fig:RightInvertible}
    Examples of transformations and right invertibility.
  }
  \begin{center}
    \includegraphics{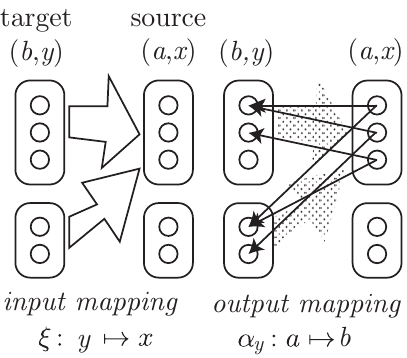}
  \end{center}
  \captionof{figure}{\label{Fig:GraphicalDeterministic}
    Graphical representation of deterministic local maps, see Figure~\ref{Fig:InvarianceC} for a complex example.
  }
}[-1cm]
We first discuss deterministic stochastic matrices, and their invertibility properties, as deterministic local transformations can be seen as their generalization.
We then complete the characterization of deterministic local transformations made in the previous sections.
In particular, we introduce three representations of local transformations: an abstract formulation using the pair of mappings $(\xi, \overline{\alpha})$, a graphical representation and a representation as a block matrix, all of which play a role in this Part~\ref{Part:Liftings} of our manuscript.
We also provide the composition rule of deterministic local transformations, and their decomposition into pure input and output maps.

\subsection{Local transformations as generalized stochastic matrices}
Consider a scenario where A and B have only one input $X=Y=1$.
Then, any local transformation $\Lmap{}: \Pspace{A} \to \Pspace{B}$ has a particularly simple form
\begin{equation}
  \forall a,b, \quad \Lambda_{b,a} \ge 0\;, \hspace{2cm} \forall a, \quad \sum_b \Lambda_{b,a} = 1\;,
\end{equation}
which corresponds to a {\em column-stochastic matrix}, not necessarily square.
In such simple scenarios, deterministic local transformations are matrices with a single coefficient equal to $1$ in each column.
We look at the cases where such a transformation can be reversed.
For that, we need a left inverse element $\Lmap{}^\text{-1L}$ such that
\begin{equation}
  \forall \vecP{A}, \qquad \vecP{A} = \Lmap{}^\text{-1L} ~ \Lmap{} ~ \vecP{A}\;,
\end{equation}
which implies that $\Lmap{}^\text{-1L} ~ \Lmap{} = \Mid$ and thus that $\Lmap{}$ is {\em left-invertible}.\
This is possible only if $\Lmap{}$ has at most one nonzero element in each row; otherwise, components of $\vecP{A}$ are mixed in a nonreversible manner.
This also implies that the number of outputs cannot decrease: $B_1 \ge A_1$.
We illustrate left-invertibility in Figure~\ref{Fig:LeftInvertible} by considering $\Lmap{}$ as the biadjacency matrix~\cite{Arumugam2016} of a bipartite edge-weighted graph.
The vertices on the right represent the output values $a=1,\ldots,A_1$, while the vertices on the left represent $b=1,\ldots,B_1$.
When the transformation is deterministic, notice that the graph represents a deterministic mapping $\alpha: \domain{A_1} \to \domain{B_1}$; and left-invertible deterministic transformations correspond to injective $\alpha$ that preserve distinctness.

The matrix $\Lmap{}$ is {\em right invertible} if there exists a right inverse element $\Lmap{}^\text{-1R}$ such that $\Lmap{} ~ \Lmap{}^\text{-1R} = \Mid$.
We easily check that right invertible $\Lmap{}$ have $\{0,1\}$ coefficients with at least one nonzero coefficient per row.
Those matrices have necessarily $A_1 \ge B_1$, and correspond to deterministic mappings with $\alpha$ surjective.
A graphical representation is given in Figure~\ref{Fig:RightInvertible}.

We see easily that left- {\em and} right-invertible transformations do not change the number of outputs $(A_1 = B_1)$ and correspond to permutation matrices.
For later use, we define {\em row-stochastic matrices} as matrices with nonnegative entries with each row summing to 1.

The present section only applies to local transformations with $X=Y=1$.
The next sections will study left and right-invertible deterministic local transformations with $X,Y>1$, but we need to complete a few definitions before that.

\subsection{Deterministic local transformations}
\label{Sec:LifitingDefinitions:DeterministicLocalTransformations}
The definition of deterministic local maps was only sketched in Section~\ref{Sec:Def:LocalMaps:Deterministic}.
We recall that $\Lmap{}: \Pspace{A} \to \Pspace{B}$ deterministic corresponds to a mapping of inputs $\xi: \domain{Y} \to \domain{X}$ and a sef of output mappings $\alpha_y: \domain{A_{\xi(y)}} \to \domain{B_y}$ such that
\begin{equation}
  \Lambda_{(b,y),(a,x)} = \begin{cases} 1 & \mbox{if }  x = \xi(y) \mbox{ and }  b = \alpha_y(a), \\ 0 & \mbox{otherwise.} \end{cases}
\end{equation}

\subsubsection{Graphical representation of deterministic local maps}

The corresponding graphical representation is shown in Figure~\ref{Fig:GraphicalDeterministic}, and already used in Figure~\ref{Fig:InvarianceC}.
A first structural level corresponds to the input mapping, which goes from the target input $y$ to the source input $x$.
Then, for each target input $y$, we have a mapping of the outputs corresponding to the source input $x=\xi(y)$.

\subsubsection{Deterministic local transformations as block matrices}
\label{Sec:DeterministicBlockMatrices}
Consider the matrix representation of a deterministic local transformations $\Lmap{}: \Pspace{A} \to \Pspace{B}$.
We split the source and target vector spaces as follows
\begin{equation}
  \Pspace{A} = \RR^{A_1} \oplus \ldots \oplus \RR^{A_X}, \qquad \Pspace{B} = \RR^{B_1} \oplus \ldots \oplus \RR^{B_Y}\;,
\end{equation}
and decompose accordingly $\Lmap{}$ as a block matrix
\begin{equation}
  \Lmap{} =
  \left(
    \begin{array}{l|l|l|l}
      H_{11} & H_{12} & \ldots & H_{1X} \\
      \hline
      H_{21} & H_{22} &        & H_{2X} \\
      \hline
      \vdots &       &        &       \\
      \hline
      H_{Y1} & H_{Y2} &        & H_{YX}
    \end{array}
  \right)\;,
  \qquad
  H_{yx} \in \RR^{B_y \times A_x}\;.
\end{equation}

Deterministic local maps impose a specific structure on the blocks $H_{yx}$.
As the deterministic $\Lmap{}$ has nonzero coefficients for input pairs $(x,y)$ such that $x = \xi(y)$, that prescribes that the nonzero blocks are $H_{y,\xi(y)}$.
When viewed as a block matrix, $\Lmap{}$ has the sparsity pattern of a row-stochastic matrix: a single nonzero element in the column $x$ for each $y$.
Now, for given $y$, the block $H_{y,\xi(y)}$ is fixed by the map $b = \alpha_y(a)$; thus $H_{y,\xi(y)}$ is a column-stochastic matrix: a single nonzero element in the row $b$ for each column $a$.

\marginnote{
  Consider the deterministic map from $\pst{A} = (3,2)$ to $\pst{B}=(2,2,2)$ given by $\xi = (1,1,2)$, $\alpha_1 = (1, 2, 2)$, $\alpha_2 = (1,1,2)$ and $\alpha_3 = (2,1)$, where we described the maps $\xi$ and $\alpha_y$ by their images.
  Then $\Lmap{}$
  \[
    \left (
      \begin{array}{@{}ccc|cc@{}}
        1 & 0 & 0   & 0 & 0 \\
        0 & 1 & 1   & 0 & 0 \\
        \hline
        1 & 1 & 0   & 0 & 0 \\
        0 & 0 & 1   & 0 & 0 \\
        \hline
        0 & 0 & 0   & 0 & 1 \\
        0 & 0 & 0   & 1 & 0 \\
      \end{array}
    \right )
  \]
  has a row-stochastic block sparsity pattern
  \[
    \begin{pmatrix}
      1 & 0 \\
      1 & 0 \\
      0 & 1
    \end{pmatrix}
  \]
  while the individual nonzero blocks are column-stochastic.
  This map factorizes as a pure input map
  \[
    \begin{pmatrix}
      \Mid & 0 \\
      \Mid & 0 \\
      0 & \Mid
    \end{pmatrix}
  \]
  followed by a pure output map with diagonal blocks
  $H_{11} = \begin{pmatrix} 1 & 0 & 0 \\ 0 & 1 & 1 \end{pmatrix}\;$,
  $H_{22} = \begin{pmatrix} 1 & 1 & 0 \\ 0 & 0 & 1 \end{pmatrix}\;$,
  and
  $H_{33} = \begin{pmatrix} 0 & 1 \\ 1 & 0 \end{pmatrix}$.
}
\subsubsection{Composition rules}
\label{Sec:LiftingDefinitions:Composition}
Consider the deterministic local maps
\begin{equation}
  \Lmap{AB}: \Pspace{A} \to \Pspace{B}, \qquad \Lmap{BC}: \Pspace{B} \to \Pspace{C}, \qquad \Lmap{AC}: \Pspace{A} \to \Pspace{C}
\end{equation}
such that $\Lmap{AC} = \Lmap{BC} ~ \Lmap{AB}$, where we used the labels A, B, C to describe transformations of the same device.
We describe those maps by the (input, output) mappings
\begin{equation}
  s = (\xi, \overline{\alpha}), \qquad t = (\psi, \overline{\beta}), \qquad u = (\zeta, \overline{\gamma})
\end{equation}
in that order, such that $u = t ~ s$.

As we have $\xi: \domain{Y} \to \domain{X}$ and $\psi: \domain{Z} \to \domain{Y}$, and should have $\zeta: \domain{Z} \to \domain{X}$, we easily deduce
\begin{equation}
  \label{Eq:CompositionXi}
  x = \zeta(z) = \xi(\psi(z)) \quad \Rightarrow \quad \zeta = \xi \circ \psi \;.
\end{equation}
For the outputs, we have $b = \alpha_y(a)$ and $c = \beta_z(b)$, thus $\gamma_z: \domain{A_{\zeta(z)}} \to \domain{C_z}$ should be
\begin{equation}
  \label{Eq:CompositionAlpha}
  c = \gamma_z(a) = (\beta_z \circ \alpha_y)(a) = (\beta_z \circ \alpha_{\psi(z)})(a) \quad \Rightarrow \quad \gamma_z = \beta_z \circ \alpha_{\psi(z)}\;.
\end{equation}

\subsubsection{Decomposition into pure input and output maps}
Among the deterministic local maps, we single {\em pure input maps} $\Lmap{}^\mathrm{I}$, where all $\alpha_y$ are identity maps, and {\em pure output maps} $\Lmap{}^\mathrm{O}$, where $\xi$ is the identity.
\begin{proposition}
  \label{Prop:DecompositionDeterministicLocalMaps}
  Any deterministic local map $\Lmap{A}: \Pspace{A} \to \Pspace{C}$ can be factored as the composition of a pure input map $\Lmap{}^\mathrm{I}: \Pspace{A} \to \Pspace{B}$ followed by a pure output map $\Lmap{}^\mathrm{O}: \Pspace{B} \to \Pspace{C}$:
  \begin{equation}
    \Lmap{A} = \Lmap{}^\mathrm{O} ~ \Lmap{}^\mathrm{I}\;.
  \end{equation}
\end{proposition}
\begin{proof}
  We construct the decomposition using the pure input map $\Lmap{}^\mathrm{I}$ that has the same input mapping $\xi$ as $\Lmap{A}$, but all output mappings equal to the identity.
  The output mappings of $\Lmap{A}$ are written $\alpha_z$ and are indexed by the target input $z$.
  Thus, $\Lmap{}^\mathrm{O}$ is given by an identity input mapping and all the output mappings of $\Lmap{A}$.
\end{proof}

\clearpage

\section{Defining liftings}
\label{Sec:DefiningLiftings}

We now motivate our definition of liftings, through the equivalency of expressions or behaviors, when using local transformations to maximize the average payoff. 

\subsection{Maximal violations of a Bell inequality}
The motivation of our study comes from the maximal violation of a Bell inequality that can be obtained for given devices of known behavior.
\begin{definition}
  \label{Def:MaximalAveragePayoff}
  The {\em maximal average payoff} of the behavior $\vecP{CD} \in \Pspace{C} \otimes \Pspace{D}$ under the Bell expression $\expr{AB} \in \left( \Pspace{A} \otimes \Pspace{B} \right)^*$ is written $\left< \expr{AB} \right>_{\vecP{CD}}^\star$ where
  \begin{equation}
    \left< \expr{AB} \right>_{\vecP{CD}}^\star = \max_{\Lmap{C}, \Lmap{D}} \expr{AB} \left( (\Lmap{C} \otimes \Lmap{D}) ~ \vecP{CD} \right ) \;,
\end{equation}
\marginnote{By convexity, it is sufficient to consider deterministic local maps instead of generic local transformations.}
and the maximization is done over all deterministic local maps $\Lmap{C}: \Pspace{C} \to \Pspace{A}$ and $\Lmap{D}: \Pspace{D} \to \Pspace{B}$.
\end{definition}

Note that under that definition, all relabelings of CHSH provide the same test for a given behavior.
Or, given a Bell inequality, all relabelings of the PR box provide the same violation.
The definition also allows testing Bell inequalities with behaviors that do not have the same cardinality, thus removing an element of arbitrariness in the process.

\subsection{Equivalent behaviors}
\label{Sec:EquivalentBehaviors}
We now precise the notion of equivalency of behaviors.
\begin{definition}
  Two behaviors $\vecP{CD} \in \Pspace{C} \otimes \Pspace{D}$ and $\vecP{EF} \in \Pspace{E} \otimes \Pspace{F}$ are {\em equivalent} if, for any Bell expression $\expr{AB} \in (\Pspace{A} \otimes \Pspace{B})^*$ of arbitrary cardinality, we have
  \begin{equation}
    \left< \expr{AB} \right>_{\vecP{CD}}^\star = \left< \expr{AB} \right>_{\vecP{EF}}^\star\;.
  \end{equation}
\end{definition}
For example, if $\vecP{CD}$ corresponds to a relabeling of the inputs and outputs of $\vecP{EF}$, then both behaviors are equivalent.
Inputs and outputs relabelings are reversible transformations.
More generally, we can speak of interconvertible transformations.
\begin{definition}
  Two behaviors $\vecP{CD} \in \Pspace{C} \otimes \Pspace{D}$ and $\vecP{EF} \in \Pspace{E} \otimes \Pspace{F}$ are interconvertible if there exists deterministic maps $\Lmap{C}$, $\Lmap{D}$, $\Lmap{E}$ and $\Lmap{F}$ (types clear from the context) such that
  \begin{equation}
    \vecP{EF} = (\Lmap{C} \otimes \Lmap{D}) ~ \vecP{CD}, \qquad \vecP{CD} = (\Lmap{E} \otimes \Lmap{F}) ~ \vecP{EF}\;.
  \end{equation}
\end{definition}
\marginnote{Forgetting about D and F (see below) and for $\pst{C}=(2,2)$, $\pst{E}=(3,2)$, a pair of interconvertible maps is given by
  \[
    \Lmap{C} = \left(
      \begin{array}{@{}ll|ll@{}}
        1 & 0 & 0 & 0 \\
        0 & 1 & 0 & 0 \\
        0 & 0 & 0 & 0 \\
        \hline
        0 & 0 & 1 & 0 \\
        0 & 0 & 0 & 1
      \end{array}
    \right)
  \]
  and
  \[
    \Lmap{E} = \left(
      \begin{array}{@{}lll|ll@{}}
        1 & 0 & 0 & 0 & 0 \\
        0 & 1 & 1 & 0 & 0 \\
        \hline
        0 & 0 & 0 & 1 & 0 \\
        0 & 0 & 0 & 0 & 1
      \end{array}
    \right)\;,
  \]
  where $\Lmap{C}$ is a fine graining of outputs, while $\Lmap{E}$ is a coarse graining of outputs.
}[-2cm]

The latter definition implies the former.
\begin{proposition}
  \label{Prop:InterconvertibilityBehaviors}
  Interconvertibility implies equivalency.
\end{proposition}
\begin{proof}
  For that, we apply the statement below in two directions. 
  Let $\expr{AB}$ be an arbitrary Bell expression and $\vecP{EF} = (\Lmap{C} \otimes \Lmap{D}) ~ \vecP{CD}$.
  Then
  \begin{equation}
    \left< \expr{AB} \right>_{\vecP{EF}}^\star \le \left< \expr{AB} \right>_{\vecP{CD}}^\star \;.
  \end{equation}
  This is easily proven from Definition~\ref{Def:MaximalAveragePayoff} and the composition of deterministic maps (Proposition~\ref{Prop:DetClosed}).
\end{proof}
We single out {\em behavior lifitings} as deterministic transformations $\Lmap{C} \otimes \Lmap{D}$ that can be applied to any $\vecP{CD}$, such that the resulting $\vecP{EF}$ can be converted back.
Due to the tensor structure, it is sufficient to consider first transformations that apply only to the first device.
We are thus looking at the deterministic maps $\Lmap{C}: \Pspace{C} \to \Pspace{E}$ with a corresponding $\Lmap{E}: \Pspace{E} \to \Pspace{C}$ such that
\begin{equation}
  ((\Lmap{E} \circ \Lmap{C}) \otimes \Mid) ~ \vecP{CD} = \vecP{CD}\;.
\end{equation}
\marginnote{The technicality concern degenerate scenarios where the device C cannot signal, and has at least two inputs, say $z=1,2$ that have only one output: $A_1 = A_2 = 1$.
  We do not get any information from the device C when using $z=1,2$, and correlations of other parties are not affected due to nonsignaling.
  Then, local transformations that affect only those inputs cannot then be distinguished from the identity.
}
Modulo a small technicality (see margin note), this implies that $\Lmap{C}$ is left invertible ($\Lmap{E} \circ \Lmap{C} = \Mid$); note that in general, the left inverse $\Lmap{E}$ is not unique.
To summarize, we find behavior liftings in the set of left invertible deterministic local maps, which are studied in Section~\ref{Sec:LiftingBehaviors}.

\subsection{Equivalent Bell expressions}
\label{Sec:EquivalentBellExpressions}
We now adapt the concepts of equivalence and interconvertibility to Bell expressions.
\begin{definition}
  \label{Def:EquivalentBellExpressions}
  Two Bell expressions $\expr{AB} \in (\Pspace{A} \otimes \Pspace{B})^*$ and $\expr{CD} \in (\Pspace{C} \otimes \Pspace{D})^*$ are {\em equivalent} if, for any behavior $\vecP{EF} \in \Pspace{E} \otimes \Pspace{F}$ of arbitrary cardinality, we have
  \begin{equation}
    \left< \expr{AB} \right>_{\vecP{EF}}^\star = \left< \expr{CD} \right>_{\vecP{EF}}^\star\;.
  \end{equation}
\end{definition}
We also derive a notion of interconvertibility of Bell expressions.
\begin{definition}
  \label{Def:InterconvertibleBellExpressions}
  The Bell expressions $\expr{AB}$ and $\expr{CD}$ are interconvertible if there exists deterministic maps $\Lmap{A}$, $\Lmap{B}$, $\Lmap{C}$ and $\Lmap{D}$, types clear from context, such that
  \begin{equation}
    \expr{CD} = \expr{AB} ~ (\Lmap{C} \otimes \Lmap{D}), \qquad \expr{AB} = \expr{CD} ~ (\Lmap{A} \otimes \Lmap{B}) \;.
  \end{equation}
\end{definition}
The equivalent of Proposition~\ref{Prop:InterconvertibilityBehaviors} is given below.
\begin{proposition}
  \label{Prop:InterconvertibilityExpressions}
  Interconvertibility of Bell expressions implies equivalency.
\end{proposition}
\begin{proof}
  Let $\vecP{EF}$ be an arbitrary behavior and $\expr{CD} = \expr{AB} ~ (\Lmap{C} \otimes \Lmap{D})$ .
  As in Proposition~\ref{Prop:InterconvertibilityBehaviors}, we have
  \begin{equation}
    \left< \expr{CD} \right>_{\vecP{EF}}^\star \le \left< \expr{AB} \right>_{\vecP{EF}}^\star \;.
  \end{equation}
\end{proof}
Thus, if $\expr{AB}$ and $\expr{CD}$ are interconvertible, they are equivalent.
Now, {\em expression liftings} are deterministic transformations $\Lmap{C} \otimes \Lmap{D}$ that apply to any $\expr{AB}$ and create an expression $\expr{CD}$ interconvertible with $\expr{AB}$.
Considering only the first device, we are looking at deterministic maps $\Lmap{C}: \Pspace{C} \to \Pspace{A}$ and $\Lmap{A}: \Pspace{A} \to \Pspace{C}$ such that
\begin{equation}
  \expr{AB} ~ \left[(\Lmap{C} \circ \Lmap{A}) \otimes \Mid\right] = \expr{AB}\;,
\end{equation}
which, modulo the same technicality as before, implies that $\Lmap{C} \circ \Lmap{A} = \Mid$ and thus that $\Lmap{C}$ is right invertible.
In general, the right inverse $\Lmap{A}$ is not unique, and we will find expression liftings among right invertible deterministic local transformations, which will be studied in Section~\ref{Sec:LiftingExpressions}.

\subsection{Liftings, reorderings, relabelings}

Some deterministic local maps have a left {\em and} right inverse.
Among those, we single out:
\marginnote{An example of a pair of reorderings is given by
  \[
    \Lmap{A} = \left(
      \begin{array}{@{}ll|lll@{}}
        0 & 0 & 1 & 0 & 0 \\
        0 & 0 & 0 & 1 & 0 \\
        0 & 0 & 0 & 0 & 1 \\
        \hline
        1 & 0 & 0 & 0 & 0 \\
        0 & 1 & 0 & 0 & 0
      \end{array}
    \right)
  \]
  and
  \[
    \Lmap{C} = \left(
      \begin{array}{@{}lll|ll@{}}
        0 & 0 & 0 & 1 & 0 \\
        0 & 0 & 0 & 0 & 1 \\
        \hline
        1 & 0 & 0 & 0 & 0 \\
        0 & 1 & 0 & 0 & 0 \\
        0 & 0 & 1 & 0 & 0
      \end{array}
    \right)\;.
  \]
}[1cm]
\begin{itemize}
\item Permutations of inputs and outputs that do not modify the cardinality $\pst{A}$.
  They are maps of the form $\Lmap{A}: \Pspace{A} \to \Pspace{A}$ and correspond to the subgroup of invertible elements in the semigroup/monoid of local deterministic maps $\Det{A}$ (see Proposition~\ref{Prop:DetClosed}).
  We name those elements {\em relabelings}; we will study them in detail in another work~\cite{RossetInPrepRelabelings}.
\item Permutations of inputs and outputs that modify the cardinality $\pst{A}$ (up to a permutation).
  They are maps of the form $\Lmap{A}: \Pspace{A} \to \Pspace{C}$ such that a unique $\Lmap{C}: \Pspace{C} \to \Pspace{A}$ has the property that
  \begin{equation}
    \Lmap{C} \circ \Lmap{A} = \Mid_\Pspace{A}, \qquad \Lmap{A} \circ \Lmap{C} = \Mid_\Pspace{C}\;.
  \end{equation}
  We name those elements {\em reorderings}.
\end{itemize}

We now consider left or right invertible maps that are neither relabelings or reorderings.
The left invertible maps are {\em behavior liftings}, while the right invertible maps are {\em expression liftings}.
We study them in detail in Section~\ref{Sec:LiftingBehaviors} and Section~\ref{Sec:LiftingExpressions}.

\clearpage

\section{Behavior liftings}
\label{Sec:LiftingBehaviors}

Recall that in Section~\ref{Sec:EquivalentBehaviors}, we defined behaviors liftings as the deterministic transformation $\Lmap{A}: \Pspace{A} \to \Pspace{C}$ that can be applied to any $\vecP{AD}$, such that the resulting $\vecP{CD} = (\Lmap{C} \otimes \Mid) ~ \vecP{AD}$ can be converted back; this implied that $\Lmap{A}$ is left invertible, i.e. there exists $\Lmap{C}$ such that
\begin{equation}
  \Lmap{C} \circ \Lmap{A} = \Mid\;.
\end{equation}
We now characterize such left invertible transforms.
\subsection{Left invertible local transformations}
The main idea is a follows.
For the map $\Lmap{A}$ to be left invertible, there must be at least one target input $z$ that corresponds to each source input $x$, with the output processing reversible.
This implies that the number of inputs $Z \ge X$.
\begin{proposition}
  \label{Prop:LiftingBehaviors}
  Any left invertible map $\Lmap{A}: \Pspace{A} \to \Pspace{C}$ can be decomposed into the composition of a map $\Lmap{}^\mathrm{L}: \Pspace{A} \to \Pspace{B}$ and a $\Lmap{}^\mathrm{R}: \Pspace{B} \to \Pspace{C}$:
  \begin{equation}
    \Lmap{A} = \Lmap{}^\mathrm{R} \circ \Lmap{}^\mathrm{L}
  \end{equation}
  such that
  \begin{itemize}
    \item The map $\Lmap{}^\mathrm{L}$ has the elongated block diagonal form (following Section~\ref{Sec:DeterministicBlockMatrices})
      \begin{equation}
        \Lmap{}^\mathrm{L} =
        \left(
          \begin{array}{c|c|c}
            L_{11} & 0 & \ldots \\ \hline
            L_{21} & 0 &   \\ \hline
            \vdots & 0 &   \\ \hline
            0 & L_{12} &  \\ \hline
            0 & L_{22} &   \\ \hline
            0 & \vdots &
          \end{array} \right)\;,
      \end{equation}
      such that each column contains one $\{L_{1x}\}$ or several $\{L_{1x}, L_{2x}, \ldots\}$ column-stochastic deterministic matrices, and the first block of each column $L_{1x}$ is left-invertible.
    \item The map $\Lmap{}^\mathrm{R}$ corresponds to a reordering/relabeling of inputs that does not transform outputs.
    \end{itemize}
\end{proposition}
\begin{proof}
We study left invertible maps $\Lmap{A}: \Pspace{A} \to \Pspace{C}$, which are characterized by an element $s$ of the form
\begin{equation}
  s = (\xi, \overline{\alpha}), \qquad \xi: z \mapsto x, \quad \alpha_z: a \mapsto c\;.
\end{equation}
along with their left inverses $\Lmap{C}: \Pspace{C} \to \Pspace{A}$, characterized by an element $t$ of the form
\begin{equation}
  t = (\zeta, \overline{\gamma}), \qquad \zeta: x \mapsto z, \quad \gamma_x: c \mapsto a \;.
\end{equation}

\marginnote{Remember that $g \circ f = \id$ implies that $f$ is injective and $g$ is surjective.}
We require $\Lmap{C} \circ \Lmap{A} = \Mid$, and thus $t \circ s = e$; for what follows, we remember the composition rules of Section~\ref{Sec:LiftingDefinitions:Composition}.

According to Eq.~\eqref{Eq:CompositionXi}, we have $\xi \circ \zeta = \id$.
Thus $\xi$ is surjective: All original inputs have to be present in $\Lmap{A} ~ \vecP{A}$.
The surjectivity of $\xi: \domain{Z} \to \domain{X}$ implies that $X \le Z$, so that the number of inputs cannot decrease.

According to Eq.~\eqref{Eq:CompositionAlpha}, we have $\gamma_x \circ \alpha_{\zeta(x)} = \id$ for all $x$; all $\gamma_x$ are surjective while those $\alpha_z$ with $z$ in the image of $\zeta$ are injective.
This implies that for each $x$, there is a $z$ in the preimage set $\xi^{-1}(x)$ such that $\alpha_z: \domain{A_x} \mapsto \domain{C_z}$ is injective.
When $\xi^{-1}(x) = \{z\}$, then the corresponding output mapping must be reversible.
When $\xi^{-1}(x) = \{z_1, z_2, \ldots \}$, the input $z$ corresponds to a source input $x$ with many clones.
One of the clones, say $z_1$, has to correspond to a reversible output mapping; it will be used to reverse the local transformation with $\zeta(x) = z_1$.
The other clones $\{z_2, \ldots\}$ can use an arbitrary output mapping.
For $z = \zeta(x)$, we have $A_x \le C_z$, so that the corresponding number of outputs cannot decrease.
The form of the proposition follows easily by doing first the cloning of inputs and their transformations, and then reordering the clones into their final place. 
\end{proof}
\subsection{Types of output transformations}
We interpret the proposition as follows.
Any diagonal behavior lifting $\Lmap{}^\mathrm{L}$ can be decomposed into a pure input transformation and a pure output transformation.
The input transformation operates on inputs by reordering and possibly cloning them.
The presence of several blocks in a column corresponds to the case of cloning inputs.
Considering now the output transformation, we distinguish several cases for the blocks $L_{ix}$.
When $i=1$, the block $L_{ix}$ must be left invertible, meaning it is of either of the following types.
\begin{enumerate}[I.]
\item The block $L_{ix}$ is square and the cardinality of outputs $C_z = A_x$ does not change.
  The block is simply a permutation matrix, and corresponds to an output relabeling.
\item The block $L_{ix}$ is not square, the cardinality of outputs $C_z > A_x$ increases.
  It corresponds to a fine-graining of outputs, along, possibly, a permutation of outputs, as in Figure~\ref{Fig:LeftInvertible}.
\end{enumerate}
For $i>1$, two additional types are possible.
\begin{enumerate}[I.]
  \setcounter{enumi}{2}
\item The block $L_{ix}$ assigns a deterministic output $c$, and corresponds to a matrix with a single row identically 1.
\item The block $L_{ix}$ of mixed type: It is neither a relabeling (I), a fine-graining (II) or a deterministic assignment (III).
\end{enumerate}

\subsection{Example}

The PR box correlations $\vecP{AD}$ are expressed in the scenario with $\pst{A} = (2,2)$ and $\pst{D} = (2,2)$ by the probability distribution
\begin{equation}
  \Pgivenop{AB}{XY}{ab}{xy} = \begin{cases} \frac{1}{2} & \text{ if } a \oplus b = x ~ y\;, \\ 0  & \text{ otherwise.} \end{cases}
\end{equation}
We now lift this behavior to the scenario with $\pst{C} = (3,3,3)$ and $\pst{D} = (2,2)$ unchanged.
For that, we enumerated the $5832 = 3^6 ~ 2^3$ deterministic local transformations $\Lmap{A}: \Pspace{A} \to \Pspace{C}$, of which $2592$ are left invertible and thus correspond to behavior liftings.
From these deterministic transformations, we obtained $1944$ unique lifted behaviors, all of which are extremal points of the nonsignaling polytope~\cite{Brunner2014} in the scenario $(3,3,3) \otimes (2,2)$.
In addition, that scenario has $3^3 \cdot 2^2 = 108$ deterministic extremal boxes.
Together the deterministic boxes and the lifted PR boxes form the $2052$ vertices of the nonsignaling polytope: Thus, that scenario does not exhibit new nonsignaling boxes.

\clearpage

\section{Lifting expressions}
\label{Sec:LiftingExpressions}

To lift expressions, we are looking at a right invertible deterministic map $\Lmap{C}: \Pspace{C} \to \Pspace{A}$, such that any Bell expression $\expr{AD}$ can be lifted into the expression $\expr{CD} = \expr{AD} ~ (\Lmap{C} \otimes \Mid)$, and the transformation is reversible, $\expr{AD} = \expr{CD} ~ (\Lmap{A} \otimes \Mid)$, because $\Lmap{A}: \Pspace{A} \to \Pspace{C}$ is the right inverse of $\Lmap{C}$.

\subsection{Right invertible local transformations}
We characterize liftings of expressions as follows.
\begin{proposition}
  \label{Prop:LiftingExpressions}
  Any right invertible map $\Lmap{C}: \Pspace{C} \to \Pspace{A}$ can be decomposed into the composition of a transformation $\Lmap{}^\mathrm{R}: \Pspace{C} \to \Pspace{B}$ and a transformation $\Lmap{}^\mathrm{L}: \Pspace{B} \to \Pspace{A}$:
  \begin{equation}
    \Lmap{C} = ~ \Lmap{}^\mathrm{L} ~ \Lmap{}^\mathrm{R}
  \end{equation}
  where
  \begin{itemize}
  \item the transformation $\Lmap{}^\mathrm{R}$ is a reordering/relabeling of inputs that does not affect outputs,
    \item the transformation $\Lmap{}^\mathrm{L}$ is a lifting and has the elongated block diagonal form
      \begin{equation}
        \Lmap{}^\mathrm{L} =
        \left(
          \begin{array}{c|c|c|c|c|c}
            L_1 & 0 & \ldots & 0 & 0 & \ldots \\ \hline
            0 & 0 &  & L_2 & 0 & \ldots \\ \hline
            \vdots & & & & &
          \end{array} \right)
      \end{equation}
      where each row contains a single nonzero block $L_z$, and the blocks $L_z$ are right-invertible.
    \end{itemize}
  \end{proposition}
  \begin{proof}
    We now characterize the set of right invertible deterministic maps $\Lmap{A}: \Pspace{A} \to \Pspace{C}$, characterized by
    \begin{equation}
      s = (\xi, \overline{\alpha}), \qquad \xi: z \mapsto x, \quad \alpha_z: a \mapsto c\;.
    \end{equation}
    along with their right inverses $\Lmap{C}: \Pspace{C} \to \Pspace{A}$, characterized by
    \begin{equation}
      t = (\zeta, \overline{\gamma}), \qquad \zeta: x \mapsto z, \quad \gamma_x: c \mapsto a \;.
    \end{equation}
    so that $\Lmap{A} \circ \Lmap{C} = \Mid_{\Pspace{C}}$, and thus $s ~ t = e$.

    According to the composition rule~\eqref{Eq:CompositionXi}, we have $\zeta \circ \xi = \id$; thus $\xi$ is injective.
    The injectivity of $\xi: \domain{Z} \to \domain{X}$ implies that $X \ge Z$.

    According to the composition rule~\eqref{Eq:CompositionAlpha}, we have $\alpha_z \circ \gamma_{\xi(z)} = \id$ for all $z$; thus the $\alpha_z$ are surjective.
    The surjectivity of $\alpha_z: \domain{A_{\xi(z)}} \to \domain{C_z}$ implies $A_{\xi(z)} \ge C_z$.

    Using Proposition~\ref{Prop:DecompositionDeterministicLocalMaps}, we write $\Lmap{A} = \Lmap{}^\mathrm{O} ~ \Lmap{}^\mathrm{I}$.
    If $\xi$ is bijective, then $\Lmap{}^\mathrm{I}$ is simply a reordering of inputs.
    If it is not, then $\Lmap{}^\mathrm{I}$ removes some inputs.
    Now, moving to $\Lmap{}^\mathrm{O}$.
    If all $\alpha_z$ are bijective, then the action of $\Lmap{}^\mathrm{O}$ is a relabeling of outputs.
    Otherwise, as $\alpha_z$ is surjective, its action corresponds to a coarse-graining of outputs.

    The interpretation using an elongated block diagonal form follows by reordering inputs separately.
  \end{proof}

In essence, the zero columns correspond to inputs that are removed by the transformation, while the output transformations correspond to coarse-graining of outputs.
Both those nonreversible transformations actually discard information; however, when these right invertible deterministic transformations act on a Bell expression, they create a (seemingly) more complicated expression, as the evaluation of that expression corresponds to: Take a behavior in a complex scenario, discard information and evaluate a Bell expression in a simpler scenario.
Note that the transformations thus identified correspond to the input and output liftings described in~\cite{Pironio2005}.

\subsection{Example: causal inequalities}
We are concerned with probability distributions~$\Pgiven{AB}{XY}$ which arise from {\em causal\/} orderings of the two parties Alice and Bob.
The assumption on causality says that either one party is in the causal past of the other, or that both parties cannot communicate~\cite{Oreshkov2012}.
Thus, such a distribution is called {\em causal\/} if it admits a decomposition
\begin{align}
	\Pgiven{AB}{XY} = p\Pgiven{A}{X}\Pgiven{B}{AXY} + (1-p)\Pgiven{A}{BXY}\Pgiven{B}{Y}
	\,,
\end{align}
where with probability~$p$ Alice can be thought of being in the past of Bob, and with probability~$1-p$ Bob is in the past of Alice (note that the nonsignaling terms can be absorbed in any of the two summands).
This decomposition induces {\em causal inequalities\/} (a Bell-like inequality) with which one can test whether a given distribution can be decomposed as above.

Our results on liftings are not only applicable to Bell inequalities that test whether a distribution is nonlocal, but also to these causal inequalities.
We continue by providing a short example.

A simple causal inequality (a facet of the causal polytope) for two parties with binary inputs and binary outputs is~\cite{Branciard2016}:
\begin{align}
	\Pr(A=Y,B=X) \leq \frac{1}{2}
	\,,
\end{align}
and is called the {\em Guess Your Neighbours Input\/} game: We are concerned with the probability that Alice guesses Bob's input and that, simultaneously, Bob guesses Alice's input.
It is easy to see that in a setup where Alice is in the past of Bob, or Bob in the past of Alice, or both parties are space-like separated, or even in any convex combination thereof, the inequality is satisfied.
To adopt our notation, we rewrite this inequality as a Bell expression~$\expr{AB} \in (\Pspace{A}\otimes\Pspace{B})^*$, where~$\phi(a,b,x,y)=1$ for~$a=y\wedge b=x$ and~$0$ otherwise.
In this notation, every causal distribution~$\Pgiven{AB}{XY}$ satisfies~$\expr{AB}\vecP{AB} \leq 2$.

We now {\em lift\/} this expression to a new scenario where Bob has one output in addition.
This is done by the use of the local map~$\Mid\otimes\Lambda$ with
\begin{align}
  \Lambda =
  \begin{pmatrix}
    1 & 0 & 0 & 0 & 0 & 0 \\
    0 & 1 & 1 & 0 & 0 & 0 \\
    0 & 0 & 0 & 1 & 0 & 0 \\
    0 & 0 & 0 & 0 & 1 & 1
  \end{pmatrix}
	\,.
\end{align}
So, we get
\begin{align}
	\expr{AB'}' = \expr{AB}(\Mid\otimes\Lambda)
	\,,
\end{align}
with
\begin{align}
	\phi'(a,b,x,y)=
	\begin{cases}
		1&\text{for }a=y\wedge b=x\\
		1&\text{for }a=y\wedge x=2 \wedge b=3\\
		0&\text{otherwise.}
	\end{cases}
\end{align}

Since this is a lifiting, it does not help to detect non-causal behaviors as compared to the initial~$2$-outputs inequality.
As a sanity check, we verified with the help of CDD~\cite{Fukuda1997} that this inequality is indeed a facet of the causal polytope for two parties with a binary input and a binary output for Alice, and with a binary input and ternary output for Bob.
This polytope consists of
\begin{itemize}
	\item 24 non-negativity facets
		\begin{align*}
			\forall a,x,y\in\{0,1\}, b\in\{0,1,2\}:
			\Pgivenop{AB}{XY}{a,b}{x,y} \geq 0
			\,,
		\end{align*}
	\item 48 {\em lazy guess your neighbours input (LGYNI)\/} facets
		\begin{align*}
			\forall& a,x,y\in\{0,1\}, b\in\{0,1,2\},\\
			\forall& c,d,e,q,r,s\in\{0,1\}: q\not=r \vee q\not= s \vee r\not=s\\
			&\Pr\Big(
			(x\oplus c)(a\oplus d \oplus y)=0\,\wedge
			( y=e \vee
			(
			(b=0 \wedge x=q_y) \vee\\
			&\quad(b=1 \wedge x=r_y)\,\vee
			(b=2 \wedge x=s_y) )
			)\Big)
			\leq \frac{3}{4}
			\,,
		\end{align*}
	\item 144 {\em guess your neighbours input (GYNI)\/} facets
		\begin{align*}
			\forall& a,x,y\in\{0,1\}, b\in\{0,1,2\},\\
			\forall& c,d,q_0,r_0,s_0,q_1,r_1,s_1\in\{0,1\}: q_i\not=r_i \vee q_i\not= s_i \vee r_i\not=s_i\\
			&\Pr\Big(
			y=c\oplus dx\oplus a\,\wedge\\
			&\quad
			(
			(b=0 \wedge x=q_y) \vee
			(b=1 \wedge x=r_y) \vee
			(b=2 \wedge x=s_y)
			)
			\Big)
			\leq \frac{1}{2}
			\,,
		\end{align*}
\end{itemize}
where we adopted the notion of Ref.~\cite{Branciard2016}: inputs and outputs are labeled starting with~$0$.
The {\em lazy guess your neighbours input\/} facets are facets where each party guesses the other party's input conditioned on her or his input.

\clearpage
\part{Conclusion}
%

John Bell's~\cite{Bell1964} work on the EPR Paradox~\cite{Einstein1935} is a milestone for the foundations of quantum theory.
His work provided a new language to study and characterize how physical objects behave: it is now common to study correlations allowed by physical theories, and their limits.
Bell inequalities are exemplary, as they formulate the boundaries of the correlations attainable in agreement with a local hidden variable model.
If we find that a Bell inequality is violated by the use of some physical system or theory, then we must conclude that the system or theory is in disagreement with at least one of the assumptions of any local hidden variably model.
This approach has been extended to incorporate other concepts or assumptions, leading to other Bell-like inequalities (see {\it e.g.}, Ref.~\cite{Oreshkov2012} in the case of definite causal order).
However, a comprehensive study of the mathematical structure of correlations and Bell inequalities was partly missing.

In this work we studied behaviors as well as Bell-like inequalities, and showed that these objects are dual to each other.
Then, we showed which local transformations can be applied to behaviors or inequalities: our answer is that local transformations correspond to preprocessing and postprocessing with memory.
We reached the same conclusion by taking two natural but different approaches: one based on causality (causes precede effects), and one inspired by quantum theory, where maps are restricted to be completely positive (local transformations must preserve normalization and nonnegativity of probability distributions, even when applied to a joint distribution).

After that, we studied the properties of these transformations.

In a geometric approach, we studied invariant subspaces of local transformations.
We showed that local transformations can be decomposed into parts representing non-signaling, signaling, and normalization constraints.
This allows a compact description of both behaviors and Bell-like inequalities, in particular when partial signaling is allowed.
We generalized a number of earlier results to scenarios where (partial) signaling is involved: the equivalence of Bell-like inequalities, the optimization of the variance of Bell statistical estimators.
We also showed how our approach translates to steering scenarios.

In the algebraic approach, we studied how local transformations compose.
We showed how the liftings of Pironio~\cite{Pironio2005} arise from invertibility of local transformations.
We generalized the known definition of liftings to apply it to behaviors, and to signaling scenarios as well.

We finish by discussing a few open questions that remain.

Since we managed to apply liftings to causal inequalities, it is natural to ask which generalizations remain to be explored.
We briefly showed how to apply our technique to steering scenarios; however, the generalization to quantum resources of arbitrary type~\cite{Schmid2019, Rosset2019} remains to be done -- for example, which liftings apply to teleportages~\cite{Cavalcanti2017a}.
Local transformations will also apply to nonlinear causal incompatibility inequalities, derived for classical hidden variables~\cite{Wolfe2019a} or quantum sources~\cite{Wolfe2019b}.
In the study of causal order, process matrices are mathematical objects in a generalization of quantum theory~\cite{Oreshkov2012} that allow for violations of causal inequalities; it remains to be seen if our approach might shed a light on the structure of those processes.

We observed that deterministic local transformations form a monoid; while we explored the representations of that monoid, its structure could be explored, as it corresponds to a generalization of the transformation monoid~\cite{Steinberg2010}.

Finally, while we know how to recognize liftings of Bell inequalities~\cite{Rosset2014a}, the question of recognizing automatically that some behavior is a lifting is still open.

\begin{acknowledgements}
  Research at Perimeter Institute is supported in part by the Government of Canada through the Department of Innovation, Science and Economic Development Canada and by the Province of Ontario through the Ministry of Economic Development, Job Creation and Trade.
  This publication was made possible through the support of a grant from the John Templeton Foundation.
  The opinions expressed in this publication are those of the authors and do not necessarily reflect theviews of the John Templeton Foundation.
\"A.B.~is supported by the Erwin Schr\"odinger Center for Quantum Science \& Technology (ESQ), and the Austrian Science Found (FWF): Z3 and F71.
N.G.~acknowledges financial support by the Swiss NCCR SwissMap.
M.-O.R. is supported by the Swiss National Fund Early Mobility Grant P2GEP2\_191444 and acknowledges the Spanish MINECO (Severo Ochoa SEV-2015-0522), Fundacio Cellex and Mir-Puig, Generalitat de Catalunya (SGR 1381 and CERCA Programme).
\end{acknowledgements}

\bibliographystyle{plainnat}
\bibliography{bellspaces}

\end{document}